\documentclass[preprint,11pt]{elsarticle}
\usepackage{amsthm}
\usepackage{amsmath}
\allowdisplaybreaks
\usepackage{minted}
\usepackage[utf8]{inputenc}
\usepackage{url}
\usepackage{adjustbox}
\newtheorem{proposition}{Proposition}
\usepackage{amssymb}
\usepackage{graphicx}
\usepackage{bm}
\usepackage{placeins}
\usepackage{rotating}
\usepackage{graphicx} 
\usepackage{lscape}
\usepackage{multirow}
\usepackage{tfrupee}
\usepackage{subcaption}
\usepackage{comment}
\usepackage[english]{babel}
\usepackage{setspace}
\usepackage{enumerate}
\usepackage{epstopdf}
\usepackage{color}
\usepackage{float}
\usepackage{scalerel,stackengine}
\usepackage{mathtools}
\usepackage{verbatim}
\usepackage{esvect}
\usepackage{placeins}
\usepackage{natbib}
\usepackage[utf8]{inputenc}
\usepackage{amsmath}
\usepackage{mathrsfs}
\usepackage{orcidlink}

\usepackage{booktabs}
\usepackage{longtable}
\usepackage{hyperref}
\usepackage[labelfont=bf, textfont=normal]{caption}
\hypersetup{colorlinks,linkcolor={blue},citecolor={blue},urlcolor={red}}  
\usepackage{algorithm}
\usepackage{color}
\usepackage{algpseudocode}
\usepackage{amsthm,setspace}
\newtheorem{theorem}{Theorem}

\newtheorem{remark}{Remark}
 \usepackage{lscape}

\usepackage[margin=1.8cm]{geometry}
\usepackage[numbers]{natbib}

\usepackage{lineno}
\makeatletter
\def\ps@pprintTitle{%
 \let\@oddhead\@empty
 \let\@evenhead\@empty
 \def\@oddfoot{\reset@font\hfil\thepage\hfil}
 \let\@evenfoot\@oddfoot
}
\makeatletter

\begin{document}

\begin{frontmatter}
\title{A Flexible Modeling of Extremes in the Presence of Inliers}
\author[label1]{Shivshankar Nila\,\orcidlink{0009-0004-0465-5490}}
\author[label1]{Ishapathik Das\corref{cor1}}
\ead{ishapathik@iittp.ac.in}
\cortext[cor1]{Corresponding author: Ishapathik Das}
\author[label1]{N. Balakrishna}

\address[label1]{Department of Mathematics and Statistics, Indian Institute of Technology Tirupati, Tirupati, Andhra
Pradesh, India} 

\doublespacing
\begin{abstract}
\noindent Many real-world phenomena, including life-testing and environmental data, exhibit positive observations with an excess of zeros (referred to as inliers at zero), which pose significant modeling challenges. The presence and proportion of such inliers may affect extreme value analysis (EVA), biased parameter estimates, and distorted tail estimation. A key and persistent challenge in EVA is selecting an appropriate threshold to accurately capture tail behaviour.
This paper proposes a flexible modeling framework that jointly accounts for extremes, inliers, and the tail fraction, treating both the threshold and the tail fraction as parameters.
 The model parameters are estimated using maximum likelihood methods. The proposed approach is compared with classical graphical tools, including the mean excess plot, parameter stability plot, and Pickands plot. 
Theoretical properties are established, and simulation studies and real-data applications demonstrate that the proposed model provides more accurate and reliable estimation, reducing bias and mean squared error and improving tail inference in the presence of inliers.
\end{abstract}

\end{frontmatter}
 \doublespacing
\section{Introduction}\label{sec1}
\noindent In statistical modeling, a primary objective is to describe the key characteristics of a population based on observed data. Standard probability distributions help capture general data patterns but are often insensitive to extreme observations with unusually large or small magnitudes.  There is increasing interest in predicting extreme events across a wide range of disciplines. In such applications, the focus lies in estimating extreme effects, that is, the maximum or minimum impact of an event, such as temperature, rainfall intensity, floods, financial losses, and lifetimes. Extreme value theory (EVT) provides a statistical framework for modeling extreme observations and quantifying their extreme effects or associated risks \cite{coles2001introduction}. In practice, the peaks-over-threshold (POT) approach is widely used in EVT, which models the exceedances above a chosen threshold. 
 A key issue in extreme value
modeling is the choice of a threshold beyond which the asymptotically motivated extreme
value models provide an adequate tail approximation. 
The threshold must be sufficiently high to ensure that the asymptotics underlying the (generalized Pareto distribution) GPD approximation are reliable, thus reducing the bias.
However, the small sample size for high thresholds increases the variance of the parameter estimates \cite{roth2016threshold}. Threshold selection involves a trade-off between bias and variance. In numerical threshold selection methods, a common choice is the $90$th percentile. Alternatively, the threshold may be selected using the $k$ largest observations, where $k = \sqrt{n}$ or $k = \frac{\sqrt[2/3]{n}}{\log(\log(n))}$ \cite{scarrott2012review}. 
Although numerical methods are simple, they are often unreliable because every dataset has its own characteristics.  Classical graphical diagnostics, such as the mean excess plot (MEP) \cite{coles2001introduction}, parameter stability plot (PSP) \cite{coles2001introduction}, Pickands plot (PP) \cite{finkenstadt2003extreme, frigessi2002dynamic}, and Hill plot \cite{finkenstadt2003extreme, frigessi2002dynamic}, are also widely used for threshold selection. However, these approaches are subjective and typically ignore threshold uncertainty.
To address threshold uncertainty, \citet{behrens2004bayesian} proposed the extreme value mixture model (EVMM), which treats the threshold as an estimable parameter. 
EVMM consists of two components: a flexible parametric model for observations below the threshold (the ``bulk model'') and a GPD for exceedances above the threshold (the ``tail model''). 
Various EVMMs with different bulk distributions have been proposed in the literature; some are listed in \hyperref[Appendix_EVMM_summary]{Appendix~A4}. For comprehensive reviews, see \cite{scarrott2012review, hu2013extreme}, and \cite{dey2016extreme}. 
Recent methods have been proposed for threshold selection \citep{sakthivel2025dual}, but they often ignore non-exceedances and the associated threshold uncertainty. Since non-exceedances may contain relevant information, neglecting them can lead to suboptimal results.
 Several studies have aimed to improve extreme-quantile and tail estimation using alternative distributions and estimation procedures, including the Gumbel and Weibull distributions \cite{pinheiro2016comparative,jokiel2024estimation}.
Recently, \citet{wu2025heavy} proposed a new estimator for the tail index of heavy-tailed distributions. Accurate tail estimation remains a central objective in such approaches, and in the present work, we address this through improved threshold selection.

Many random events occur in real-life situations; for example, in life-testing experiments, an item might fail immediately, resulting in a lifetime of zero. In rainfall data, zero observations may represent dry days during the rainy season, contributing to drought conditions. In these cases, inliers often emerge as a subset of observations that appear inconsistent with the rest of the dataset. Many real-world data sets exhibit a high concentration of observations near a particular point \( x_{0} \), while the remaining responses follow varying distributions. We refer to such observations as "early failures" or  "inliers" near the point \( x_{0} \). When \( x_{0} = 0 \) and observations at particular \( x_{0} = 0 \), this is specifically termed an "instantaneous failure" or an  "inliers at zero" .
To address such scenarios, ~\citet{muralidharan2006analysis} proposed a model incorporating the Weibull distribution and a degenerate distribution at zero. More generally, the Tweedie framework has been used to model data with such characteristics \cite{kurz2017tweedie,svensson2017statistical}.
\subsection{Effect of Inliers on Classical Threshold Estimation: Motivating Example}
\noindent To demonstrate the impact of inliers on classical threshold estimation methods, we examine parameter estimates obtained from graphical tools such as the MEP \cite{coles2001introduction}, PSP \cite{coles2001introduction}, and PP \cite{finkenstadt2003extreme, frigessi2002dynamic}.
We consider a dataset characterized by a known threshold $u = 8.04719$ and a GPD tail with shape parameter $\xi=0.20$ and scale $\sigma=5$, along with inliers at the origin. We then analyze two versions of this dataset: one including inliers, denoted by the full dataset $\mathbf{x} = \{x_1, x_2, \dots, x_n\}$ with $x_i \ge 0$, and the other excluding inliers, denoted by $\mathbf{x}^{+} = \{x_i : x_i > 0\}$, which retains only the nonzero observations. Figure~\ref{fig:meanexcess} shows that threshold estimation may be distorted in the presence of inliers, affecting the estimates of the GPD parameters. This highlights the sensitivity of classical graphical methods to the presence of inliers and suggests that ignoring or mishandling them can lead to biased threshold and parameter estimates, potentially resulting in unreliable tail inference. This motivates the need for a more robust modeling framework for extremes in the presence of inliers.
\vspace{-0.2cm}
\begin{figure}[!htbp]
    \centering
    \includegraphics[width=0.92\textwidth]{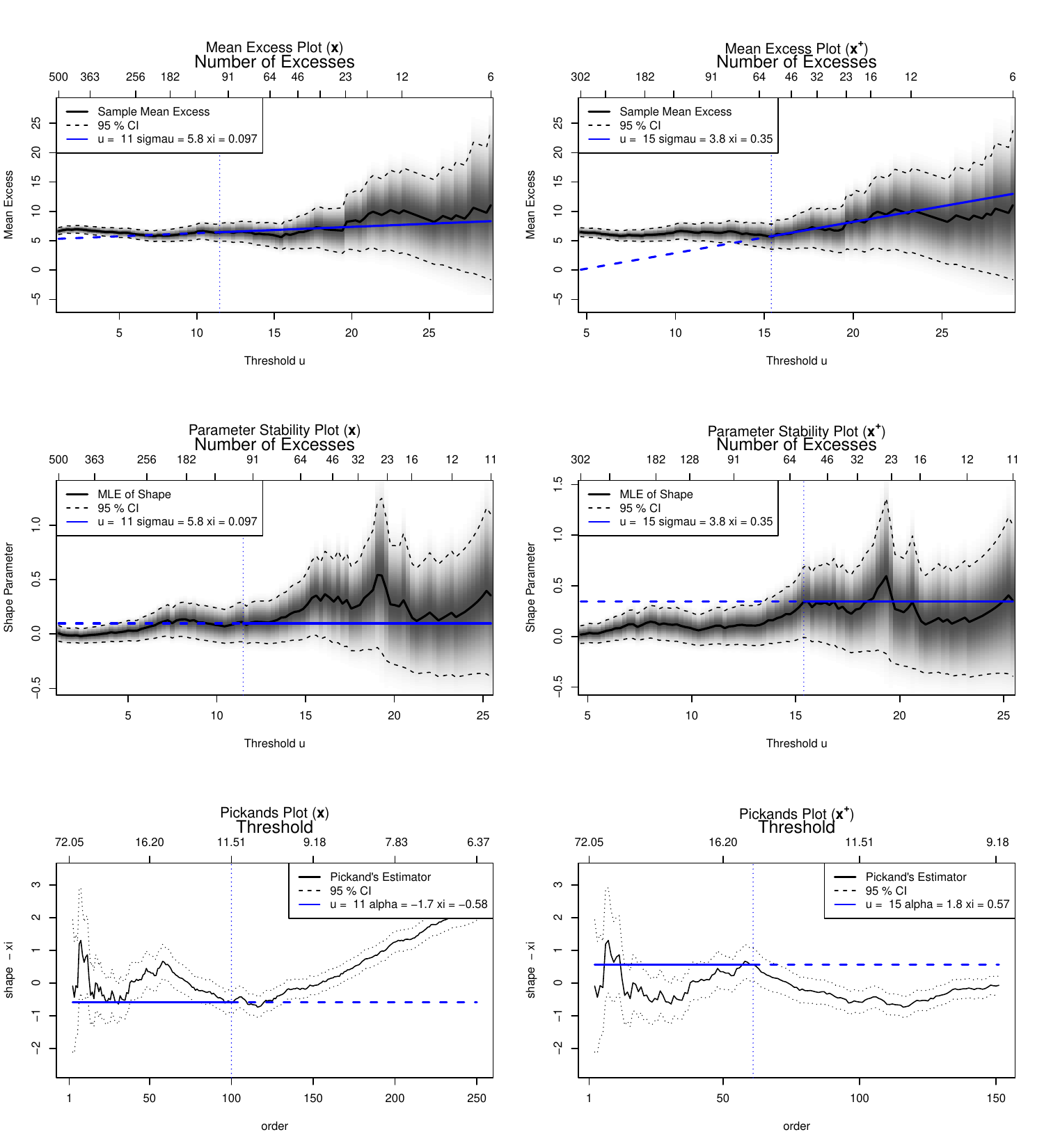}
    \caption{Comparison of MEP, PSP, and PP for the full dataset including inliers ($\mathbf{x}$) and the reduced dataset excluding inliers ($\mathbf{x}^{+}$). The results show that the presence of inliers affects threshold selection and the corresponding GPD parameter estimates ($\sigma$, $\xi$).}
    \label{fig:meanexcess}
\end{figure}
\subsection{Motivation and Research Gap}
\noindent 
In extreme-value mixture modeling, it is crucial to accurately characterise the tail while simultaneously determining the threshold and tail fraction. This becomes more difficult in the presence of inliers, as their proportion may influence threshold selection and the estimation of extreme-value parameters. 
This may distort estimates of risk, extreme quantiles, and other tail-related measures.
 Although inliers play an important role in the data, single-tail EVMMs with positive support are often ignored or not thoroughly studied across the full support, leading to suboptimal inference. There is scope for further research, as no unique and generalized framework for defining these EVMMs, particularly regarding the tail fraction (e.g., $\phi_u$), has been studied in detail \cite{hu2013extreme,dey2016extreme}. Moreover, misspecification of the distribution below the threshold can affect the tail fraction, threshold estimate, and extreme-parameter value estimates \cite{hu2018evmix}.
The robustness of both bulk and tail components remains a concern, and no single mixture specification performs uniformly well across different datasets \citep{scarrott2012review,marambakuyana2024composite}. These challenges have been noted in the literature, emphasizing the need for flexible and robust modeling frameworks \citep{behrens2004bayesian, hu2018evmix}.
Another limitation of EVT is its focus on exceedances, although non-exceedances may also contribute to extreme outcomes; for example, persistent moderate rainfall may cause flooding, while prolonged moderate heat may contribute to drought and water scarcity. Therefore, modeling the full distribution, including both extreme and non-extreme observations, is important \cite{andre2024joint}.

The main goal of this study is to determine whether the choice of observation form below the threshold affects the estimation of the threshold and extreme-value parameters. We propose a novel flexible extreme value inlier mixture model (FEVIMM), which combines a degenerate distribution at the origin, a bulk distribution for values below the threshold (excluding the origin), and a GPD above the threshold, with the tail proportion treated as a parameter. We evaluate robustness, goodness-of-fit, and parameter estimation accuracy by comparing the proposed FEVIMM with the established EVMM \cite{behrens2004bayesian} and flexible EVMM (FEVMM) \cite{macdonald2011flexible}.
Our analysis assesses whether the proposed model improves tail estimation, enhances the precision of threshold determination, and reduces bias and mean squared error (MSE) in parameter estimation. 
Model performance is assessed using standard goodness-of-fit tests, including the Anderson-Darling ($AD$), Cramér-von Mises ($CvM$), and Kolmogorov-Smirnov ($KS$) tests.

The paper is structured as follows. Section \ref{background} provides the background and key concepts, including EVT and EVMM. Section~\ref{methodology} presents the proposed methodology, including the model formulation, model properties, risk measures, the data-simulation algorithm, parameter estimation, and the asymptotic properties of the maximum likelihood estimator (MLE). Section \ref{SStudy} provides an extensive Monte Carlo simulation study. 
Section \ref{arStudy2} presents an application of the proposed methodology to real data. Finally, Section~\ref{conclusion} concludes with a discussion of the findings and possible extensions for future research. 
The proofs of the propositions, the elements of the asymptotic dispersion matrix, and additional simulation results are provided in the appendices.

\section{Background and Preliminaries} \label{background}
  \subsection{Extreme Value Theory} \label{evt}
 \noindent A popular approach in EVT is the POT method, which considers all data points exceeding a threshold and models the exceedances, enabling a more detailed and accurate analysis of tail behaviour. This methodology is supported by a key result from \citet{Pickands1975}. A sequence of independent and identically distributed (\text{i.i.d.}) observations $\{x_i : i = 1, \ldots, n\}$, under certain mild conditions, the excesses $x - u$ over some suitably high threshold $u$ can be well approximated by a GPD (denoted by $\text{GPD}(\sigma_u, \xi)$), as showed by \citet{davison1990models}, and defined as
\begin{equation*}
G(x|\xi,\sigma_u,u)= \Pr(X \leq x \mid X > u) =
\left\{
\begin{array}{ll}
1-\left(1+\xi\frac{x-u}{\sigma_u}\right)^{-1/\xi},& \mbox{if } \xi\neq 0,\\
1-\exp\left(-\frac{x-u}{\sigma_u}\right), & \mbox{if }\xi=0,
\end{array}
\right.
\end{equation*}
for $\xi\in\mathbb{R}$ and $\sigma_u > 0
$, where the support is $x \geq u$ if $\xi \geq 0$ and $u \leq x \leq u-\sigma_u/\xi$ if $\xi<0$. Hence, the GPD is bounded if $\xi<0$ and unbounded from above if $\xi\geq 0$. In some applications, a third implicit parameter is required, the probability of being above the threshold, called the ``tail fraction'', denoted as $\phi_u = \Pr(X > u)$. This quantity is required in calculating quantities like the unconditional survival probability as
\begin{equation}
\Pr(X > x) = \phi_u \left[1 - \Pr(X \leq x \mid X > u)\right].
\label{eq:unconditional_survival}
\end{equation}
 \noindent The implicit parameter $\phi_u$ is estimated using the MLE, which is simply the sample proportion above the threshold. The choice of threshold for fitting the GPD is crucial and arbitrary, with inference varying significantly for different thresholds \cite{tancredi2006accounting}.
\subsection{Extreme Value Mixture Model} \label{evmm}
 \noindent In traditional threshold selection methods, once the threshold is known and fixed, the GPD parameters are estimated based on this assumption, thereby ignoring the uncertainty associated with threshold selection. In contrast, \citet{behrens2004bayesian} introduced an EVMM that treats the threshold as a parameter and explicitly accounts for uncertainty in its estimation.
 The cumulative distribution function (CDF) $\mathcal{F}$ of an EVMM is defined  as
\begin{equation}
\label{eq:myotherequation7}
\mathcal{F}(x|\eta^*, u, \xi, \sigma) = \begin{cases}
    H(x|\eta^*), & \text{if } x < u, \\
    H(u|\eta^*) + [1 - H(u|\eta^*)] \cdot G(x|\Psi), & \text{if } x \geq u,
\end{cases}
\end{equation}
where $H$ is the CDF, parametrized by $\eta^*$, of the bulk, i.e., the portion of data below the threshold $u$ and $G$ is the CDF of GPD with parameters $\Psi=\{\xi,\sigma,u\}$, which models the tail of the distribution above the threshold $u$.  A variety of approaches for threshold estimation, including EVMMs, were reviewed by \citet{scarrott2012review}; some of them, summarized in \hyperref[Appendix_EVMM_summary]{Appendix~A4}, that use the full dataset and do not require a fixed threshold.
\section{Proposed Methodology} \label{methodology}
  \subsection{Proposed Model} \label{sec3}
\noindent In this section, we propose a novel flexible extreme value inlier mixture model (FEVIMM), which jointly models the degenerate distribution, bulk, and tail of the data. The model avoids threshold selection issues by treating the threshold as an estimable parameter and accounting for the associated uncertainty. Observations falling below the threshold \( u \) are assumed to follow a particular distribution, either a degenerate distribution or a bulk distribution \(G^*(\cdot \mid \Phi)\) parameterized by $\Phi$, while those exceeding the threshold \( u \) are modeled using a GPD, \( \text{G}(\cdot \mid u, \sigma, \xi) \). The CDF \( F \) for any observation \( X \) from the FEVIMM is expressed as
\begin{equation}
\label{evimmcdf}
F(x \mid \phi_1, \Phi, u, \xi, \sigma, \phi_2) =
\begin{cases} 
0, & \text{if } x < 0, \\
\phi_1, & \text{if } x = 0, \\
\phi_1 + \left( 1 - (\phi_1 + \phi_2) \right) \frac{G^*(x \mid \Phi)}{ G^*(u \mid \Phi)}, & \text{if } 0 < x < u, \\
(1 - \phi_2) + \phi_2  G(x \mid u, \sigma, \xi), & \text{if } x \geq u.
\end{cases}
\end{equation}
where one can easily verify that the function $F(x \mid \phi_1, \Phi, u, \xi, \sigma, \phi_2)$ defines a valid CDF, since $0 < \phi_1, \phi_2 < 1$,~ $\phi_1 + (1-\phi_1-\phi_2) + \phi_2 = 1$, with $0 < \phi_1 + \phi_2 < 1 \Rightarrow 0 \le F(x) \le 1$, 
$x_1 \le x_2 \Rightarrow F(x_1) \le F(x_2)$, 
$\lim_{x \to 0^-} F(x)=0,\ \lim_{x \to \infty} F(x)=1$, and $F(x)$ is right-continuous, i.e., $\lim_{h \downarrow 0} F(x+h) = F(x)$ for all $x$.
The degenerate distribution at the origin captures inliers with a probability mass \( \phi_1 \).
The model defined in \eqref{evimmcdf} is general and, for practical use, requires a specific choice of the bulk distribution \( G^*(\cdot \mid \Phi) \), which can be modeled using gamma, Weibull, beta, or lognormal distributions, here \( G^*(\cdot \mid \Phi) \) represent the CDF of the bulk distribution. Throughout this study, we model the bulk using the gamma distribution \( G^*(x \mid \eta, \beta) \), and the tail follows a GPD \( G(x \mid \xi, \sigma, u) \) with a tail fraction parameter \( \phi_2 \).
The DF of the FEVIMM is defined as $x \geq 0$,
\begin{equation}\label{evmmlpdf1}
f(x \mid \phi_1, \Phi, u, \xi, \sigma,\phi_2) =
\begin{cases} 
\phi_1, & \text{if } x = 0, \\[2pt]
\left( 1 - (\phi_1 + \phi_2) \right) \frac{g^*(x \mid \Phi)}{ G^*(u \mid  \Phi) \,}, & \text{if } 0 < x < u, \\[2pt]
\phi_2 ~ g(x \mid u, \sigma, \xi), & \text{if } x \geq u.
\end{cases}
\end{equation}
We can express the model given in \eqref{evmmlpdf1} with a mixture model of three components as
\[
f(x) = \phi_1 \cdot \delta_0(x) + (1 - \phi_1 - \phi_2) \cdot f_1(x) + \phi_2 \cdot f_2(x),
\]
where
\[
f_1(x) = \frac{g^*(x \mid \Phi)}{G^*(u \mid \Phi)} \cdot \mathbb{I}_{(0, u)}(x), \quad
f_2(x) = g(x \mid u, \sigma, \xi) \cdot \mathbb{I}_{[u, \infty)}(x),
\]
and \( \delta_0(x) \) represents the Dirac delta function at zero, corresponding to the degenerate distribution at the origin. The \( g^*(\cdot \mid \Phi) \) represents the density function (DF) corresponding to the CDF \( G^*(\cdot \mid \Phi) \), while \( g(\cdot \mid u,\sigma,\xi) \) represents the DF corresponding to the CDF \( G(\cdot \mid u,\sigma,\xi) \).

We study key properties of the proposed model, including continuity and differentiability at $u$, as well as related risk measures.
\begin{proposition}[Continuity at the threshold $u$]
\label{rem:Continuity}
The DF \eqref{evmmlpdf1} of the FEVIMM is continuous at the threshold $u$,
i.e., $f(u^{-}) = f(u^{+})$, 
if and only if
\[
\sigma
=
\frac{\phi_2}{1-\phi_1-\phi_2}
\;\cdot\;
\frac{G^*(u \mid \Phi)}{g^*(u \mid \Phi)}.
\]
\noindent This condition ensures that the transition between the bulk and the GPD tail is continuous at $u$.
\end{proposition}
\begin{proposition}[Differentiability at the threshold $u$]
\label{rem:differentiability}
The DF \eqref{evmmlpdf1} of the FEVIMM is differentiable at the threshold $u$, 
i.e., $f'(u^{-}) = f'(u^{+})$, 
if and only if both Proposition~\ref{rem:Continuity} holds and
\begin{equation}
\label{eq:differentiability_general}
\sigma^{2}
=
\frac{\phi_{2}(1+\xi)}{\phi_{1}+\phi_{2}-1}
\;\cdot\;
\frac{G^{*}(u \mid \Phi)}{g^{*\,\prime}(u \mid \Phi)},
\end{equation}
where $g^{*\,\prime}(\cdot \mid \Phi)$ denotes the derivative of $g^{*}(\cdot \mid \Phi)$ with respect to $x$.
\medskip
\noindent
In the special case where the bulk $G^*(\cdot \mid \Phi)$ is a gamma distribution with shape and scale parameters $\eta$ and $\beta$, respectively, the condition in \eqref{eq:differentiability_general} reduces to
\begin{equation}
\label{eq:gamma_differentiability}
\sigma
=
\frac{1+\xi}{
\dfrac{1}{\beta}\;-\;\dfrac{\eta-1}{\,u\,}},
\qquad
\sigma > 0,
~~
\dfrac{1}{\beta} - \dfrac{\eta-1}{u} \neq 0.
\end{equation}
This condition ensures the DF $f(x)$ is differentiable at $u$. 
\end{proposition}
\subsection{Risk Measures}\label{risk}
\noindent\noindent Risk measures are tools used in risk analysis to quantify the level of risk associated with a random variable by summarizing it into a single value. Value-at-Risk (VaR) and Tail-Value-at-Risk (TVaR) are among the most widely used risk measures. Several risk measures have been developed; for more details, see \cite{szego2004risk}. It was developed to deal with heavy-tailed distributions when the information above a given threshold
is required. In this section, we derive the VaR and TVaR for the proposed FEVIMM.
\begin{proposition}[Risk measures: VaR]
\label{prop:quantile}
Assume that $G^*(\cdot\mid\Phi)$ and $G(\cdot\mid u,\sigma,\xi)$ admit inverse functions.
For a probability level $p\in(0,1)$, the VaR, denoted by $\mathrm{VaR}_p$, is defined as
$\mathrm{VaR}_p = \inf\{x: F(x)\ge p\}$.
Then the $\mathrm{VaR}_p$ of \eqref{evimmcdf} is given by
\[
\mathrm{VaR}_p=
\begin{cases}
\displaystyle
0, 
& 0 \le p \le \phi_1, \\[6pt]
\displaystyle
G^{*-1}\!\Bigl(
\dfrac{p-\phi_1}{\,1-(\phi_1+\phi_2)\,}\,
G^*(u\mid\Phi)
\Bigr),
& \phi_1 < p \le 1-\phi_2, \\[10pt]
\displaystyle
G^{-1}\!\Bigl(
\dfrac{p-(1-\phi_2)}{\phi_2}
\Bigr),
& 1-\phi_2 < p \le 1.
\end{cases}
\]
\end{proposition}
\begin{proposition}[Risk measures: TVaR]
\label{prop:TVaRquantile}
For a given probability level $p$, \text{let}~ $\mathrm{VaR}_p$ be defined as in Proposition~\ref{prop:quantile}. 
The \emph{TVaR}, also called 'expected shortfall' or 'conditional tail expectation', is defined for a random variable $X$ as
\[
\mathrm{TVaR}_p(X)
= \mathbb{E}[X \mid X > \mathrm{VaR}_p]
= \frac{\mathbb{E}\!\left[ X \,\mathbf{1}_{\{X > \mathrm{VaR}_p\}} \right]}
       {\mathbb{P}(X > \mathrm{VaR}_p)}
= \frac{1}{1-p} \int_{\mathrm{VaR}_p}^{\infty} x\, f(x)\,dx,
\]
where $f(x)$ is the density in \eqref{evmmlpdf1}. The TVaR at level~$p$ can therefore be expressed as
\[
\mathrm{TVaR}_p =
\begin{cases}
\displaystyle
\frac{1}{1-\phi_1}\Biggl[
\int_{0^+}^{u} x \,
\frac{1-(\phi_1+\phi_2)}{G^*(u\mid\Phi)}\,g^*(x\mid\Phi)\,dx
\;+\;
\int_{u}^{\infty} x\, \phi_2\, g(x\mid u,\sigma,\xi)\,dx
\Biggr], 
& 0 \le p \le \phi_1,\\[10pt]
\displaystyle
\frac{1}{1-p}\Biggl[
\int_{\mathrm{VaR}_p}^{u} x \,
\frac{1-(\phi_1+\phi_2)}{G^*(u\mid\Phi)}\,g^*(x\mid\Phi)\,dx
\;+\;
\int_{u}^{\infty} x\, \phi_2\, g(x\mid u,\sigma,\xi)\,dx
\Biggr], 
& \phi_1 < p \le 1-\phi_2,\\[10pt]
\displaystyle
\frac{1}{1-p}\,
\int_{\mathrm{VaR}_p}^{\infty} x \,\phi_2\, g(x\mid u,\sigma,\xi)\,dx, 
& 1-\phi_2 < p \le 1.
\end{cases}
\]
\noindent
Note the lower limit $0^+$ in the first case excludes the possible atom at $x=0$ of mass $\phi_1$; when $p\le\phi_1$ we have $\mathrm{VaR}_p=0$.
Both VaR and TVaR are widely used risk measures in risk analysis. For the FEVIMM, these quantities can be obtained analytically through the derived integrals or estimated numerically using a Monte Carlo simulation.
\end{proposition}

The proofs of Propositions~\ref{rem:Continuity} and \ref{rem:differentiability} are provided in \hyperref[Appendix_Proof_prepositions]{Appendix~A1}, whereas Propositions~\ref{prop:quantile} and \ref{prop:TVaRquantile} follow directly from the definitions of the corresponding risk measures \cite{szego2004risk}.
\subsection{Algorithm for Data Simulation}\label{Simulation}
\noindent This section provides an algorithm for simulating data $x_i$, 
$i = 1, 2, 3, \dots, n$, from the proposed model, FEVIMM
$F(x \mid \phi_1, \Phi, u, \xi, \sigma, \phi_2)$ as defined in (\ref{evimmcdf}). The simulation procedure is presented in Algorithm~\ref{alg:simulation}.
\begin{algorithm}[!htbp]
\caption{Simulation Algorithm for the Flexible Extreme Value Inlier Mixture Model}
\label{alg:simulation}
\begin{algorithmic}[1]
\small
\Require $n$ (number of samples), $\phi_1, \Phi, u, \xi, \sigma, \phi_2$
\Ensure A sample $\{ x_1, x_2, \dots, x_n \}$ from the FEVIMM
\State Initialize an empty vector $x$ of length $n$
\For{$i = 1$ to $n$}
    \State Generate $U \sim \text{Uniform}(0,1)$
    \If{$U \leq \phi_1$} 
        \State Set $x_i \gets 0$
    \ElsIf{$\phi_1 < U \leq 1 - \phi_2$} 
        \State Compute $k \gets \frac{U - \phi_1}{1 - (\phi_1 + \phi_2)}$
        \State Compute $G_u^* \gets G^*(u \mid \Phi)$ \hfill \Comment{CDF of $G^*(\cdot \mid \Phi)$ at $u$}
        \State Set $x_i \gets Q_{G^*}(k \cdot G_u^* \mid \Phi)$ \hfill \Comment{Quantile function of $G^*(\cdot \mid \Phi)$}
    \Else
        \State Compute $c_u \gets \frac{U - (1 - \phi_2)}{\phi_2}$
        \If{$\xi \neq 0$} 
            \State Set $x_i \gets u + \frac{\sigma}{\xi} \left[ (1 - c_u)^{-\xi} - 1 \right]$ \Comment{GPD quantile ($\xi \neq 0$)}
        \Else
            \State Set $x_i \gets u - \sigma \log(1 - c_u)$ \Comment{GPD quantile ($\xi = 0$)}
        \EndIf
    \EndIf
\EndFor
\State \Return $\{ x_1, x_2, \dots, x_n \}$
\end{algorithmic}
\end{algorithm}
\subsection{Parameter Estimation} \label{estimation}\label{mle}
\noindent This section focuses on estimating the parameters of the FEVIMM. Let 
\(\bm{\Theta} = (\phi_1, \eta, \beta, u, \xi, \sigma, \phi_2)\) represent the parameters characterizing the observed data distribution. We obtain the maximum likelihood estimate by maximizing the likelihood function based on a random sample of size \( n \), denoted by \(\bm{x} = (x_1, x_2, \ldots, x_n)\).
To facilitate the likelihood formulation, we define the following indicator functions
\[
I_1(x_i) = \begin{cases} 
1 & \text{if } x_i = 0, \\
0 & \text{otherwise},
\end{cases} \quad
I_2(x_i; u) = \begin{cases} 
1 & \text{if } 0 < x_i < u, \\
0 & \text{otherwise},
\end{cases} ~~\text{and}~~ 
I_3(x_i; u) = \begin{cases} 
1 & \text{if } x_i \geq u, \\
0 & \text{otherwise}.
\end{cases}
\]
These indicator functions partition the sample space into three mutually exclusive regions corresponding to the point mass at zero, the bulk component, and the tail component.  For a random sample \( \bm{x} \) from a distribution \( F(x; \bm{\Theta}) \) (\ref{evimmcdf}) with a bulk as gamma distribution with corresponding DF \(f(x_i; \bm{\Theta})\) (\ref{evmmlpdf1}). The likelihood function and the corresponding log-likelihood are defined, respectively, as\\
$\mathcal{L}(\bm{x};\bm{\Theta})=\prod_{i=1}^{n}f(x_i;\Theta),
~~\ell(\bm{\Theta}\mid\bm{x}):=\log\mathcal{L}(\bm{\Theta}\mid\bm{x})$.
\begin{equation}  
\begin{aligned}  
\ell(\bm{\Theta} \mid \bm{x}) 
&= \sum_{i=1}^{n} I_1(x_i) \log \phi_1   
+ \sum_{i=1}^{n} I_2(x_i; u) \Big[ \log (1 - \phi_1 - \phi_2) + \log g^*(x_i \mid \eta,\beta) - \log G^*(u \mid \eta,\beta) \Big] \\  
&\quad + \sum_{i=1}^{n} I_3(x_i; u) \Big[ \log \phi_2 + \log g(x_i \mid u, \sigma, \xi) \Big],  
\end{aligned}  
\label{eq:log_likelihood}  
\end{equation}
\begin{equation}
\text{where}~~g^*(x_i\mid\eta,\beta)
=\frac{x_i^{\eta-1}e^{-x_i/\beta}}{\Gamma(\eta)\beta^\eta},
\qquad
g(x_i\mid u,\sigma,\xi)
=
\begin{cases}
\displaystyle \frac{1}{\sigma}\left(1+\xi\frac{x_i-u}{\sigma}\right)^{-(1/\xi+1)}, & \xi\neq0,\\[3pt]
\displaystyle \frac{1}{\sigma}\exp\left(-\frac{x_i-u}{\sigma}\right), & \xi=0.
\end{cases}
\label{eq:gpd_density}
\end{equation}
\noindent Here, $G^*(\cdot\mid\eta,\beta)$ is the CDF corresponding to the DF $g^*(\cdot\mid\eta,\beta)$. The parameters satisfy $0 < \phi_1, \phi_2 < 1$, $\eta > 0$, $\beta > 0$, $u > 0$, $\sigma > 0$, and $\xi \in \mathbb{R}$. We use the classical approach of estimating parameters by maximizing the log-likelihood function. The likelihood function of the FEVIMM is often complex and lacks a closed-form solution, requiring numerical methods to compute the MLEs.

  For MLE estimation, direct likelihood optimization for such mixture models is often challenging due to the likelihood surface's multimodality, especially for the threshold parameter.  Another challenge is that optimization is very sensitive to the initial threshold values.  Still, various advanced optimization tools could be used (e.g., multiple starting values or a stochastic optimizer) to overcome this multimodality.
  One commonly used approach is to profile likelihood estimation to avoid the direct optimization of the complete likelihood.
The profile likelihood over threshold ($ u$) may exhibit multimodality, so black-box optimization is similarly problematic.
A better and strongly recommended alternative is to perform a grid search over potential thresholds, which is typically used to find the value that maximises the profile likelihood.
 However, these methods do not provide uncertainty quantification for the estimated threshold. For more details, see \cite {dey2016extreme}. In contrast, approaches based on the complete likelihood, where the threshold is estimated jointly with the other model parameters, allow for uncertainty quantification. The complete likelihood approach is adopted in \cite{hu2018evmix, hu2013extreme}, and in the \texttt{evmix} package in \texttt{R} when accounting for uncertainty is required.  Consistent with this approach, we also used a similar strategy and selected initial parameter values; the detailed estimation procedure, including both the profile likelihood (grid search) and complete likelihood optimization, is given in Algorithm~\ref{alg:FEVIMM_estimation}.
Optimisation is performed using the Nelder-Mead algorithm via the \texttt{optim} function from the \texttt{stats} package in \textsf{R} \cite{team2020ra} up to 20{,}000 iterations.
\begin{remark}
Note that a key motivation for developing EVMM is to provide a more objective approach to threshold estimation, together with uncertainty quantification, compared to traditional diagnostics (\cite{dey2016extreme}). The estimates from the complete likelihood optimization $\hat{\Theta}^{(C)}$ (see Algorithm~\ref{alg:FEVIMM_estimation}, Step~7) were cross-validated against those from the profile likelihood and grid search methods, $\hat{\Theta}^{(P)}$, and were found to be largely consistent. An additional advantage of the complete optimization approach is that it allows uncertainty quantification, including the threshold parameter estimation. Consequently, all results presented in Sections \ref{SStudy} and \ref{arStudy2} are based on the complete MLE~ $\hat{\Theta}^{(C)}$.
\end{remark}
\begin{algorithm}[!htbp]
\caption{Estimation Algorithm for the Flexible Extreme Value Inlier Mixture Model}
\label{alg:FEVIMM_estimation}
\begin{algorithmic}[1]
\small
\State \textbf{Step 1.}
Given a sample $(x_1,\ldots,x_n)$, define
\[
\mathcal{X}_0 = \{x_i : x_i = 0\}, \quad
\mathcal{X}_B(u) = \{x_i : 0 < x_i < u\}, \quad
\mathcal{X}_T(u) = \{x_i : x_i \geq u\}.
\]
Consider a set of candidate thresholds $\mathcal{U} = \{u_1,\ldots,u_L\}$ chosen from high empirical quantiles, ensuring \(\sum_{i=1}^n \mathbf{1}\{x_i \geq u\} \geq m_0\),where $m_0$ denotes a pre-specified minimum number of exceedances.
\State \textbf{Step 2.}
For each $u \in \mathcal{U}$, compute the initial parameter values,
\[
\phi_1^{(0)} = \frac{1}{n} \sum_{i=1}^n \mathbf{1}\{x_i = 0\}, \quad
\phi_2^{(0)}(u) = \frac{1}{n} \sum_{i=1}^n \mathbf{1}\{x_i \geq u\}.
\]
Obtain $\Phi^{(0)} = (\eta^{(0)}, \beta^{(0)})$ from $\mathcal{X}_B(u)$ using moment- or MLE-based estimators and $(\xi^{(0)}, \sigma^{(0)})$ by fitting a GPD to $\mathcal{X}_T(u)$.
\State \textbf{Step 3.}
For each $u \in \mathcal{U}$, evaluate the profile log-likelihood
\[
\ell_p(u) = \max_{(\phi_1,\Phi,\xi,\sigma,\phi_2)}
\ell(\phi_1,\Phi,\xi,\sigma,\phi_2; u),
\]
using the initial values obtained in Step 2.
\State \textbf{Step 4.}
Obtain a threshold estimate $\hat u$ by maximizing the profile log-likelihood over $\mathcal{U}$, i.e.,
\[
\hat u = \arg\max_{u \in \mathcal{U}} \, \ell_p(u).
\]
\State \textbf{Step 5.} Given $\hat u$, update the initial parameter values by repeating Step~2 at $u=\hat u$, based on $\mathcal{X}_B(\hat u)$ and $\mathcal{X}_T(\hat u)$.

\State \textbf{Step 6.}
Given the profile likelihood estimate $\hat u$, estimate the remaining parameters 
$(\phi_1,\Phi,\xi,\sigma,\phi_2)$,
\[
(\hat\phi_1(\hat u),\hat\Phi(\hat u),\hat\xi(\hat u),\hat\sigma(\hat u),\hat\phi_2(\hat u))
=
\arg\max_{(\phi_1,\Phi,\xi,\sigma,\phi_2)}
\ell(\phi_1,\Phi,\xi,\sigma,\phi_2;\,\hat u),
\]
using initial values from Step $5$. The resulting profile likelihood-based estimator of $\Theta$ is
\[
\hat{\Theta}^{(P)} = (\hat\phi_1(\hat u),\hat\Phi(\hat u), \hat u, \hat\xi(\hat u),\hat\sigma(\hat u),\hat\phi_2(\hat u)).
\]
\State \textbf{Step 7.}
Let $\Theta = (\phi_1,\Phi,u,\xi,\sigma,\phi_2)$. Obtain the MLE $\hat{\Theta}$ by maximizing the complete log-likelihood over all parameters, using appropriate initial values, such as those obtained in Step~6 or other reasonable initializations.
\[
\hat{\Theta}^{(C)}=(\hat\phi_1,\hat\Phi,\hat u,\hat\xi,\hat\sigma,\hat\phi_2)
=
\arg\max_{(\phi_1,\Phi,u,\xi,\sigma,\phi_2)}
\ell(\phi_1,\Phi,u,\xi,\sigma,\phi_2).
\]
\end{algorithmic}
\end{algorithm}
\subsection{Asymptotic Distribution of the MLE} 
\label{asymptoticdistribution} 
\noindent 
For the study of the asymptotic properties, we consider the log-likelihood for a single observation (\(n=1\)) from~\eqref{eq:log_likelihood}, where the DF of  \(g^*(x\mid\eta,\beta)\) and the  \( g^*(x\cdot \mid \Phi) \)  are defined in~\eqref{eq:gpd_density}. The Fisher information entries and asymptotic properties are presented in this section, while the corresponding score functions and related derivations are provided in \hyperref[Appendix_asymptoticdistribution]{Appendix~A2}.

\begin{remark}[Threshold assumption]
For the development of the standard asymptotic theory of the MLE, the threshold parameter \(u\) is assumed to be known and fixed at \(u=u_0\). Treating \(u\) as an unknown parameter violates the regularity conditions required for establishing the asymptotic distribution of the MLE. Accordingly, the score functions and Fisher information matrix are derived under this assumption. In practice,  \(u\) can be estimated using a profile likelihood over a grid search; see Algorithm~\ref{alg:FEVIMM_estimation}, Step~$4$. 
Once the threshold is selected and fixed  \(u=u_0\), the remaining parameters are estimated (see Steps~$5$-$6$), which is strongly recommended \cite{hu2018evmix}, thereby ensuring the validity of asymptotic inference conditional on the selected threshold.
\end{remark}
\begin{remark}[Regularity conditions for the GPD shape parameter (\(\xi\))]
The asymptotic normality of the MLE of \(\xi\) holds under suitable regularity conditions, including finite Fisher information. The MLE of the $\xi$ does not satisfy the regularity conditions when $\xi \in (-1,-0.5)$; moreover, MLEs do not exist for $\xi < -1$, as the likelihood function can become unbounded (\cite{castillo1997fitting,do2012semiparametric}).
It is common practice to assume that $\xi$ lies within the interval $-0.5 < \xi < 0.5$, since this range is frequently encountered in practical applications as noted in~\citet{hosking1985estimation}.
\end{remark}
\subsubsection{Fisher Information and Asymptotic Normality: Case \(\xi \ne 0\)}
\noindent
Let \(\boldsymbol{\theta}=(\phi_1,\eta,\beta,\xi,\sigma,\phi_2)\) denote the parameter vector for the single-observation log-likelihood function \(\ell(\boldsymbol{\theta}\mid x)\) obtained from equation~\eqref{eq:log_likelihood}.\\
\noindent
\textbf{Notation and definitions:} \\
\noindent
Define \(h(\eta,\beta):=\log G^*(u_0\mid\eta,\beta)\), where \(G^*(u_0\mid\eta,\beta)\) is the gamma CDF evaluated at the known threshold \(u_0>0\) and acts as the normalising constant for the gamma distribution truncated above at \(u_0\), i.e., on \((0,u_0)\). Its second-order partial derivatives are denoted by
\(h_{\eta\eta}:=\frac{\partial^2 h(\eta,\beta)}{\partial\eta^2}\),
\(h_{\beta\beta}:=\frac{\partial^2 h(\eta,\beta)}{\partial\beta^2}\), and
\(h_{\eta\beta}:=\frac{\partial^2 h(\eta,\beta)}{\partial\eta\,\partial\beta}\).
The lower incomplete gamma and trigamma functions are defined, respectively, by
\(\gamma(s,x):=\int_0^x t^{s-1}e^{-t}\,dt,\ s>0,\ x\geq0\), and
\(\psi'(\eta):=\frac{d^2}{d\eta^2}\log\Gamma(\eta)\).
For \(X \geq u_{0}\), \(X \sim \mathrm{GPD}(\xi, \sigma, u_{0})\) then the standardized variable \(Z = \frac{X - u_{0}}{\sigma} \sim \mathrm{GPD}(\xi, 1, 0)\). The \( \mathbb{E}_Z[\cdot] \) denotes the expectation taken with respect to the variable \( Z \).
These notations will be used throughout the section for derivatives and likelihood calculations. 
Let \(\mathcal{I}(\boldsymbol{\theta})\) denote the Fisher information matrix, with
\(\mathcal{I}(\boldsymbol{\theta})=-\mathbb{E}\!\left[\frac{\partial^2\ell(\boldsymbol{\theta}\mid\bm{x})}{\partial\boldsymbol{\theta}\,\partial\boldsymbol{\theta}^T}\right]\)
and
\(\mathcal{I}_{ij}(\boldsymbol{\theta})=-\mathbb{E}\!\left[\frac{\partial^2\ell(\boldsymbol{\theta}\mid\bm{x})}{\partial\theta_i\,\partial\theta_j}\right]\).
where $\theta_i \in \boldsymbol{\theta}$. 
 Thus, the Fisher information entries are
\begin{align*}
\mathcal{I}_{\phi_1 \phi_1} &= \frac{1}{\phi_1} + \frac{1}{(1 - \phi_1 - \phi_2)}, \quad
\mathcal{I}_{\phi_2 \phi_2} = \frac{1}{\phi_2} + \frac{1}{(1 - \phi_1 - \phi_2)}, \quad
\mathcal{I}_{\phi_1 \phi_2} = \frac{1}{(1 - \phi_1 - \phi_2)}, \\
\mathcal{I}_{\phi_1 \eta} &= \mathcal{I}_{\phi_1 \beta} = \mathcal{I}_{\phi_1 \xi} = \mathcal{I}_{\phi_1 \sigma} =
\mathcal{I}_{\phi_2 \eta} = \mathcal{I}_{\phi_2 \beta} = \mathcal{I}_{\phi_2 \xi} = \mathcal{I}_{\phi_2 \sigma} = 0,\\
\mathcal{I}_{\eta\eta} &= (1 - \phi_1 - \phi_2) \left[ \psi'(\eta) + h_{\eta \eta} \right], \quad \mathcal{I}_{\eta \beta}  = (1 - \phi_1 - \phi_2) \left( \frac{1}{\beta} + h_{\eta \beta} \right).
\end{align*}

\begin{align*} 
\mathcal{I}_{\beta \beta} &= (1-\phi_1-\phi_2)\left[ h_{\beta \beta} + \frac{2}{\beta^2} \cdot \frac{\gamma\left(\eta + 1, \frac{u_0}{\beta}\right)}{\gamma\left(\eta, \frac{u_0}{\beta}\right)} - \frac{\eta}{\beta^2} \right], \quad 
\mathcal{I}_{\sigma \sigma} = \frac{\phi_2}{\sigma^2} \left[ 
\left(\frac{1}{\xi} + 1\right) \mathbb{E}_Z \left( 
\frac{\xi Z (2 + \xi Z)}{(1 + \xi Z)^2} 
\right) - 1 
\right]. \\ 
\mathcal{I}_{\xi\xi} 
&= 
\phi_2\,\mathbb{E}_Z\!\left[ 
\frac{2}{\xi^3}\log(1+\xi Z) 
-\frac{2Z}{\xi^2(1+\xi Z)} 
-\left(\frac{1}{\xi}+1\right)\frac{Z^2}{(1+\xi Z)^2} 
\right], 
~ 
\mathcal{I}_{\sigma\xi} 
= 
\phi_2\,\mathbb{E}_Z\!\left[ 
\frac{Z}{\sigma\xi(1+\xi Z)} 
-\left(\frac{1}{\xi}+1\right)\frac{Z}{\sigma(1+\xi Z)^2} 
\right].
\end{align*}
\begin{theorem}[Asymptotic normality of the MLE (for known threshold \( u_0 > 0 \), \(\xi\neq 0\))]
Let \( X_1, X_2, \dots, X_n \) be  \text{i.i.d.} observations from the FEVIMM with known threshold \( u_0 > 0 \), and parameter vector
\[
\bm{\theta} = (\phi_1, \eta, \beta, \xi, \sigma, \phi_2)^\top \in (0,1) \times (0, \infty)^2 \times \left( (-0.5, \infty) \setminus \{0\} \right) \times (0, \infty) \times (0,1),
\]
with \(0<\phi_1,\phi_2<1\) and \(\phi_1+\phi_2<1\). Assume that the true parameter \( \bm{\theta}_0 \) lies in the interior of the parameter space and that the regularity conditions for maximum likelihood estimation hold. Let \( \widehat{\bm{\theta}}_n \) be the MLE based on the sample of size \( n \). Then, as \( n \to \infty \),
\[
\sqrt{n} \left( \widehat{\bm{\theta}}_n - \bm{\theta}_0 \right) 
\xrightarrow{d} \mathcal{N}_6 \left( \bm{0}, \, \mathcal{I}^{-1}(\bm{\theta}_0) \right),
\]
where $\mathcal{I}(\bm{\theta}_0)$ is the Fisher information matrix based on a single observation, evaluated at the true parameter \( \bm{\theta}_0 \), with entries $\mathcal{I}_{ij}(\bm{\theta}_0)$.
The Fisher information matrix is given by
\[
\mathcal{I}(\bm{\theta}) =
\begin{bmatrix}
\mathcal{I}_{\phi_1\phi_1} & 0 & 0 & 0 & 0 & \mathcal{I}_{\phi_1\phi_2} \\[4pt]
0 & \mathcal{I}_{\eta\eta} & \mathcal{I}_{\eta\beta} & 0 & 0 & 0 \\[4pt]
0 & \mathcal{I}_{\beta\eta} & \mathcal{I}_{\beta\beta} & 0 & 0 & 0 \\[4pt]
0 & 0 & 0 & \mathcal{I}_{\xi\xi} & \mathcal{I}_{\xi\sigma} & 0 \\[4pt]
0 & 0 & 0 & \mathcal{I}_{\sigma\xi} & \mathcal{I}_{\sigma\sigma} & 0 \\[4pt]
\mathcal{I}_{\phi_2\phi_1} & 0 & 0 & 0 & 0 & \mathcal{I}_{\phi_2\phi_2}
\end{bmatrix}.
\]
\end{theorem}
\subsubsection{Fisher Information and Asymptotic Normality: Case \(\xi = 0\)}
\noindent
Let \(\boldsymbol{\theta}=(\phi_1,\eta,\beta,\sigma,\phi_2)\) denote the parameter vector for the single-observation log-likelihood function \(\ell(\boldsymbol{\theta}\mid x)\) obtained from~\eqref{eq:log_likelihood}. For \(\xi=0\), the GPD tail reduces to an exponential distribution with mean \(\sigma\). The Fisher information entries involving \(\sigma\) are
$
\mathcal{I}_{\sigma\sigma}=\frac{\phi_2}{\sigma^2},
\quad
\mathcal{I}_{\phi_1\sigma}
=\mathcal{I}_{\phi_2\sigma}
=\mathcal{I}_{\eta\sigma}
=\mathcal{I}_{\beta\sigma}=0,
$
and the remaining Fisher information entries are as in the case \(\xi\ne0\), except for those involving \(\xi\).
\begin{theorem}[Asymptotic normality of the MLE (for known threshold \( u_0 > 0 \), \(\xi=0\))]
Let \( X_1, X_2, \dots, X_n \) be \text{i.i.d.} observations from the FEVIMM with known threshold \( u_0 > 0 \), and parameter vector
\[
\bm{\theta} = (\phi_1, \eta, \beta, \sigma, \phi_2)^\top
\in (0,1)\times(0,\infty)^2\times(0,\infty)\times(0,1),
\]
with \(0<\phi_1,\phi_2<1\) and \(\phi_1+\phi_2<1\). Assume that the true parameter \( \bm{\theta}_0 \) lies in the interior of the parameter space and that the regularity conditions for maximum likelihood estimation hold. Let \( \widehat{\bm{\theta}}_n \) be the MLE based on the sample of size \( n \). Then, as \( n \to \infty \),
\[
\sqrt{n} \left( \widehat{\bm{\theta}}_n - \bm{\theta}_0 \right) 
\xrightarrow{d} \mathcal{N}_5 \left( \bm{0}, \, \mathcal{I}^{-1}(\bm{\theta}_0) \right).
\]
where $\mathcal{I}(\bm{\theta}_0)$ is the Fisher information matrix based on a single observation, evaluated at the true parameter \( \bm{\theta}_0 \), with entries
$\mathcal{I}_{ij}(\bm{\theta}_0)$.
The Fisher information matrix is defined similarly as in the case \(\xi \neq 0\), with the row and column corresponding to \(\xi\) removed.
\end{theorem}

\begin{remark}[Asymptotic confidence intervals for the MLE]
Based on the asymptotic normality of the MLE, the approximate \( 100(1 - \alpha)\% \) confidence interval for each parameter 
\( \theta_i \in \{\phi_1, \eta, \beta, \xi, \sigma, \phi_2\} \) 
is given by
\[
\widehat{\theta}_i \pm z_{\alpha/2} \sqrt{ \frac{ \left[ \mathcal{I}^{-1}(\widehat{\boldsymbol{\theta}}) \right]_{ii} }{n} },
\]
where \( z_{\alpha/2} \) represents the upper \( \alpha/2 \)-quantile of the standard normal distribution, and \( \left[ \mathcal{I}^{-1}(\widehat{\boldsymbol{\theta}}) \right]_{ii} \) denotes the \( i \)-th diagonal element of the inverse Fisher information matrix evaluated at the MLE \( \widehat{\boldsymbol{\theta}}_n \). In the case \( \xi = 0 \), the confidence interval is defined similarly for 
\( \theta_i \in \{\phi_1, \eta, \beta, \sigma, \phi_2\} \).
\end{remark}

\begin{remark}
The mixture proportions $\phi_1$ and $\phi_2$ exhibit statistical dependence in the Fisher information matrix, as indicated by the nonzero off-diagonal term $\mathcal{I}_{\phi_1\phi_2} \neq 0$. This arises from the unit-sum constraint $\phi_1 + \phi_2 + (1 - \phi_1 - \phi_2) = 1$, which imposes the relationship $\mathrm{Cov}(\widehat{\phi}_1, \widehat{\phi}_2) < 0$ under the MLE's asymptotic distribution. Consequently, increments in one proportion necessarily reduce the available probability mass for the other, leading to asymptotic correlation between $\widehat{\phi}_1$ and $\widehat{\phi}_2$. In the context of extreme value mixture modeling, explicitly accounting for the proportion of inliers ($\phi_1$) is important, as the existing EVMM literature has often either ignored this component or not accounted for it thoroughly, potentially affecting inference on the extreme component ($\phi_2$) and related risk measures.
\end{remark}

\section{Monte Carlo Simulation Study}\label{SStudy}\label{SStudy1}
\noindent In this section, we present a comprehensive Monte Carlo simulation study to evaluate the effectiveness of the proposed model. We consider different scenarios, including heavy-tailed, light-tailed, and exponentially-tailed cases, as well as varying proportions of tail data. Samples \( x_i \), for \( i = 1, 2, \dots, n \), are generated from the FEVIMM~\eqref{evimmcdf} with true parameters: \( \phi_1 \in \{0.1, 0.2, 0.3, 0.4\} \), \( \phi_2 \in \{0.1, 0.15\} \), \( \xi \in \{-0.2, 0, 0.2\} \), \( \eta \in \{1, 4\} \), \( \beta \in \{1, 5\} \), \( \sigma \in \{4, 5\} \), and threshold \( u \in \{6.6807, 9.4846, 11.5129, 19.448\} \), with sample sizes \( n \in \{150, 200, 300, 400, 500, 750, 1000\} \).
For each sample size $n$ considered, we performed 2000 Monte Carlo replications ($N$) using the procedure outlined in Section~\ref{Simulation}. MLEs of the parameter vector \( \Theta = (\phi_1, \eta, \beta, u, \xi, \sigma, \phi_2) \) are computed using the method in Section~\ref{estimation}. We assess estimator performance using the sample mean, bias, MSE, bootstrap standard error (BSE), 95\% bootstrap confidence interval (BCI), and asymptotic confidence-interval coverage probabilities.
The bias and MSE are calculated as
\begin{equation}
\mathrm{Bias}(\hat{\theta})
= \frac{1}{N} \sum_{i=1}^{N} \bigl(\hat{\theta}^{(i)} - \theta\bigr), \quad
\mathrm{MSE}(\hat{\theta})
= \frac{1}{N} \sum_{i=1}^{N}
\bigl(\hat{\theta}^{(i)} - \theta\bigr)^2,
\end{equation}
where $\theta \in \Theta$, \( \hat{\theta} \) is the estimator of \( \theta \), and $\hat{\theta}^{(i)}$ denotes the estimate of $\theta$ from the $i$th simulated dataset.

The results presented in Table~\ref {Table-1}--\ref{Table-3} indicate that the proposed methodology yields small bias and MSE, with estimated parameter values close to the true values and coverage probabilities aligning with the nominal confidence levels across all sample sizes. The parameters $u$, $\xi$ and $\sigma$ exhibit some variability at smaller sample sizes. As the sample size increases, both bias and MSE decrease, indicating improved estimation accuracy and greater estimator stability. The values shown in parentheses under the \textbf{Sample Mean}, \textbf{MSE}, and \textbf{Bias} columns are estimates obtained from the FEVMM \cite{macdonald2011flexible}. A comparison with these values shows that ignoring inliers leads to larger bias and MSE, particularly for the tail-related parameters, highlighting the importance of accounting for inliers. The \texttt{evmix} R package, developed by \citet{hu2018evmix}, was employed for this analysis and is available on CRAN. These estimates are included only as a reference to show how the FEVMM behaves when the data-generating process contains inliers. We report the results of FEVIMM, FEVMM, and EVMM, focusing on bias and MSE, and present the comparison in Section~\ref{smstudyc}. Additional results are given in the Appendix.

We compared the estimates by the MEP, PSP, and PP to both versions of the simulated dataset: one including inliers (i.e., the whole dataset $\mathbf{x}$) and the other excluding inliers (i.e., $\mathbf{x}^{+}$, where only the non-zero values are retained). For reference, the MEP, PSP, and PP were generated using the \texttt{evmix} package \cite{hu2018evmix}, and all results are based on a sample size of $n=1000$. The estimates from MEP, PSP, and PP are presented in Tables~\ref{Table-MEP1}, \ref{Table-MEP2}, and  \ref{Table-MEP3} and show that they are sensitive to inliers, with estimates varying depending on their inclusion. Estimates of the parameters $u$, $\xi$, and $\sigma$ vary across cases and may deviate from their true values, indicating unreliable estimates of the population tail. In contrast, the FEVIMM estimates remain closer to the true values and exhibit more stable behaviour.

\begin{landscape} 
\begin{table*}[htbp]
\caption{Monte Carlo simulation results ($N=2000$) for sample sizes $n = 150, 200, 300, 400,$ and $500$, with an inlier proportion of 40\% and a heavy-tailed GPD component with tail fraction 10\%, reporting sample mean, BSE, 95\% BCI, MSE, Bias, and CP.}
    \centering
    \scriptsize 
    \setlength{\tabcolsep}{3pt} 
    \renewcommand{\arraystretch}{1.3} 
    \begin{tabular}{|p{2.5cm}|p{1.8cm}|p{1.8cm}|p{2.4cm}|p{2cm}|p{3cm}|p{2.4cm}|p{2.3cm}|p{1.8cm}|} 
        \hline
       \textbf{$n$} & \textbf{Parameter} & \textbf{True} & \textbf{Sample Mean} & \textbf{BSE} & \textbf{BCI} & \textbf{MSE} & \textbf{Bias} & \textbf{CP} \\ \hline
         \multirow{6}{*}{ 150 } & $\phi_1$ &  0.4  &  0.4014  (NA) &  0.053  &  (0.2718, 0.4878)  &  0.0017  (NA) &  0.0014  (NA) &  97.2  \\ \cline{2-9}
                               & $\eta$   &  1  &  1.0733  ( 0.9899 ) &  0.3223  &  (0.7453, 1.9971)  &  0.0407  ( 0.0345 ) &  0.0733  ( -0.0101 ) &  99.6  \\ \cline{2-9}
                               & $\beta$  &  5  &  4.9232  ( 6.6769 ) &  4.8201  &  (1.9024, 16.8447)  &  4.0173  ( 14.9 ) &  -0.0768  ( 1.6769 ) &  99.8  \\ \cline{2-9}
                               & $u$      &  11.51293  &  11.2748  ( 14.1025 ) &  2.0904  &  (7.8559, 15.7551)  &  2.8547  ( 14.7699 ) &  -0.2381  ( 2.5896 ) &  92.8  \\ \cline{2-9}
                               & $\xi$    &  0.2  &  0.0353  ( 0.3213 ) &  5.2948  &  (-1.4304, 13.9537)  &  0.2846  ( 0.818 ) &  -0.1647  ( 0.1213 ) &  100  \\ \cline{2-9}
                               & $\sigma$ &  5  &  6.201  ( 6.2538 ) &  4.3969  &  (0, 13.6761)  &  14.9051  ( 39.7777 ) &  1.201  ( 1.2538 ) &  95  \\ \cline{2-9}
                               & $\phi_2$ &  0.1  &  0.1038  ( 0.1139 ) &  0.0259  &  (0.0583, 0.1692)  &  1e-04  ( 0.003 ) &  0.0038  ( 0.0139 ) &  98.8  \\ \hline

          \multirow{6}{*}{ 200 } & $\phi_1$ &  0.4  &  0.4018  (NA) &  0.0391  &  (0.337, 0.4931)  &  0.0012  (NA) &  0.0018  (NA) &  97.8  \\ \cline{2-9}
                               & $\eta$   &  1  &  1.0643  ( 0.9862 ) &  0.1973  &  (0.7777, 1.5535)  &  0.0305  ( 0.0256 ) &  0.0643  ( -0.0138 ) &  99.4  \\ \cline{2-9}
                               & $\beta$  &  5  &  4.827  ( 6.4291 ) &  4.5364  &  (2.3914, 11.7833)  &  2.4932  ( 9.9507 ) &  -0.173  ( 1.4291 ) &  100  \\ \cline{2-9}
                               & $u$      &  11.51293  &  11.288  ( 14.2069 ) &  0.9791  &  (7.4964, 11.487)  &  2.0468  ( 16.8521 ) &  -0.2249  ( 2.694 ) &  95.6  \\ \cline{2-9}
                               & $\xi$    &  0.2  &  0.0847  ( 0.2062 ) &  1.5187  &  (-0.2797, 5.4876)  &  0.1786  ( 0.5406 ) &  -0.1153  ( 0.0062 ) &  100  \\ \cline{2-9}
                               & $\sigma$ &  5  &  5.8022  ( 6.3743 ) &  1.876  &  (0.4122, 6.5289)  &  9.146  ( 33.0265 ) &  0.8022  ( 1.3743 ) &  98  \\ \cline{2-9}
                               & $\phi_2$ &  0.1  &  0.1036  ( 0.1119 ) &  0.0171  &  (0.0753, 0.1433)  &  1e-04  ( 0.0019 ) &  0.0036  ( 0.0119 ) &  99.4  \\ \hline                   
         \multirow{6}{*}{ 300 } & $\phi_1$ &  0.4  &  0.4007  (NA) &  0.0351  &  (0.3428, 0.4822)  &  8e-04  (NA) &  7e-04  (NA) &  98.8  \\ \cline{2-9}
                               & $\eta$   &  1  &  1.0551  ( 0.98 ) &  0.1803  &  (0.7586, 1.4704)  &  0.0206  ( 0.017 ) &  0.0551  ( -0.02 ) &  99.8  \\ \cline{2-9}
                               & $\beta$  &  5  &  4.734  ( 6.175 ) &  5.2947  &  (2.6529, 13.4997)  &  1.6077  ( 5.9934 ) &  -0.266  ( 1.175 ) &  99.6  \\ \cline{2-9}
                               & $u$      &  11.51293  &  11.3269  ( 14.1647 ) &  1.1467  &  (8.4679, 13.003)  &  1.3087  ( 14.7626 ) &  -0.1861  ( 2.6518 ) &  97.4  \\ \cline{2-9}
                               & $\xi$    &  0.2  &  0.1348  ( 0.163 ) &  1.7954  &  (-0.3065, 6.5139)  &  0.0947  ( 0.3353 ) &  -0.0652  ( -0.037 ) &  100  \\ \cline{2-9}
                               & $\sigma$ &  5  &  5.3717  ( 6.251 ) &  2.3184  &  (0.5083, 8.8321)  &  4.1284  ( 30.2196 ) &  0.3717  ( 1.251 ) &  98.4  \\ \cline{2-9}
                               & $\phi_2$ &  0.1  &  0.1033  ( 0.1118 ) &  0.0178  &  (0.0706, 0.1425)  &  1e-04  ( 0.0019 ) &  0.0033  ( 0.0118 ) &  99.8  \\ \hline
                               
      \multirow{6}{*}{ 400 } & $\phi_1$ &  0.4  &  0.4006  (NA) &  0.0321  &  (0.3666, 0.4939)  &  6e-04  (NA) &  6e-04  (NA) &  98.2  \\ \cline{2-9}
                               & $\eta$   &  1  &  1.0489  ( 0.9741 ) &  0.1959  &  (0.8448, 1.6278)  &  0.0155  ( 0.0133 ) &  0.0489  ( -0.0259 ) &  99.8  \\ \cline{2-9}
                               & $\beta$  &  5  &  4.7365  ( 6.1054 ) &  4.5211  &  (2.3965, 11.8848)  &  1.2963  ( 4.8553 ) &  -0.2635  ( 1.1054 ) &  100  \\ \cline{2-9}
                               & $u$      &  11.51293  &  11.3613  ( 14.2158 ) &  1.058  &  (8.8512, 13.0158)  &  1.0459  ( 17.1118 ) &  -0.1516  ( 2.7028 ) &  97.6  \\ \cline{2-9}
                               & $\xi$    &  0.2  &  0.162  ( 0.1685 ) &  1.5336  &  (-0.3573, 5.5673)  &  0.062  ( 0.2448 ) &  -0.038  ( -0.0315 ) &  100  \\ \cline{2-9}
                               & $\sigma$ &  5  &  5.1929  ( 5.9732 ) &  2.1288  &  (0.5685, 7.6824)  &  2.5714  ( 17.4148 ) &  0.1929  ( 0.9732 ) &  98.8  \\ \cline{2-9}
                               & $\phi_2$ &  0.1  &  0.1029  ( 0.1123 ) &  0.0186  &  (0.068, 0.1438)  &  1e-04  ( 0.0021 ) &  0.0029  ( 0.0123 ) &  99.6  \\ \hline
                               
        \multirow{6}{*}{ 500 } & $\phi_1$ &  0.4  &  0.4005  (NA) &  0.03  &  (0.3691, 0.4887)  &  5e-04  (NA) &  5e-04  (NA) &  98.8  \\ \cline{2-9}
                               & $\eta$   &  1  &  1.0458  ( 0.9718 ) &  0.1624  &  (0.7955, 1.442)  &  0.0131  ( 0.011 ) &  0.0458  ( -0.0282 ) &  100  \\ \cline{2-9}
                               & $\beta$  &  5  &  4.713  ( 6.0371 ) &  3.7174  &  (2.7281, 12.1438)  &  1.0375  ( 3.8358 ) &  -0.287  ( 1.0371 ) &  99.8  \\ \cline{2-9}
                               & $u$      &  11.51293  &  11.3612  ( 14.2047 ) &  1.231  &  (8.6151, 13.0803)  &  0.8638  ( 19.7881 ) &  -0.1518  ( 2.6918 ) &  97.4  \\ \cline{2-9}
                               & $\xi$    &  0.2  &  0.1762  ( 0.1689 ) &  1.4359  &  (-0.3678, 5.6121)  &  0.0449  ( 0.1662 ) &  -0.0238  ( -0.0311 ) &  100  \\ \cline{2-9}
                               & $\sigma$ &  5  &  5.0991  ( 5.7776 ) &  2.0339  &  (0.6512, 7.87)  &  1.8293  ( 11.5872 ) &  0.0991  ( 0.7776 ) &  98.2  \\ \cline{2-9}
                               & $\phi_2$ &  0.1  &  0.1031  ( 0.1136 ) &  0.0197  &  (0.0606, 0.1434)  &  1e-04  ( 0.0026 ) &  0.0031  ( 0.0136 ) &  99.6  \\ \hline

    \end{tabular}
\captionsetup{
  labelfont=bf, 
  textfont=normal 
}
    \label{Table-1}
\end{table*}
\end{landscape}

\begin{landscape} 
\begin{table*}[htbp]
\caption{Monte Carlo simulation results ($N=2000$) for sample sizes $n = 150, 200, 300, 400,$ and $500$, with an inlier proportion of 40\% and a light-tailed GPD component with a tail fraction of 10\%, reporting sample mean, BSE, 95\% BCI, MSE, Bias, and CP.}
    \centering
    \scriptsize 
    \setlength{\tabcolsep}{3pt} 
    \renewcommand{\arraystretch}{1.3} 
    \begin{tabular}{|p{2.5cm}|p{1.8cm}|p{1.8cm}|p{2.4cm}|p{2cm}|p{3cm}|p{2.4cm}|p{2.3cm}|p{1.8cm}|} 
        \hline
         \textbf{$n$}  & \textbf{Parameter} & \textbf{True} & \textbf{Sample Mean} & \textbf{BSE} & \textbf{BCI} & \textbf{MSE} & \textbf{Bias} & \textbf{CP} \\ \hline
        \multirow{6}{*}{ 150 } & $\phi_1$ &  0.4  &  0.4016  (NA) &  0.0509  &  (0.3113, 0.5139)  &  0.0017  (NA) &  0.0016  (NA) &  97.4  \\ \cline{2-9}
                               & $\eta$   &  1  &  1.0902  ( 0.9931 ) &  0.4219  &  (0.8887, 2.5004)  &  0.0426  ( 0.0313 ) &  0.0902  ( -0.0069 ) &  99.8  \\ \cline{2-9}
                               & $\beta$  &  5  &  4.6895  ( 6.5396 ) &  5.215  &  (1.4398, 13.1438)  &  2.5259  ( 12.4919 ) &  -0.3105  ( 1.5396 ) &  99.8  \\ \cline{2-9}
                               & $u$      &  11.51293  &  11.2448  ( 13.8795 ) &  1.4661  &  (6.8412, 12.8229)  &  2.6209  ( 9.5811 ) &  -0.2682  ( 2.3666 ) &  93.8  \\ \cline{2-9}
                               & $\xi$    &  -0.2  &  -0.4207  ( -0.0291 ) &  4.5609  &  (-1.3193, 11.9106)  &  0.2487  ( 0.5243 ) &  -0.2207  ( 0.1709 ) &  100  \\ \cline{2-9}
                               & $\sigma$ &  5  &  6.4367  ( 4.1139 ) &  3.6248  &  (0.0277, 11.6902)  &  12.8092  ( 11.9673 ) &  1.4367  ( -0.8861 ) &  96.8  \\ \cline{2-9}
                               & $\phi_2$ &  0.1  &  0.1036  ( 0.1116 ) &  0.0239  &  (0.0631, 0.1638)  &  2e-04  ( 0.0019 ) &  0.0036  ( 0.0116 ) &  99.6  \\ \hline
        \multirow{6}{*}{ 200 } & $\phi_1$ &  0.4  &  0.4014  (NA) &  0.0415  &  (0.3186, 0.4816)  &  0.0012  (NA) &  0.0014  (NA) &  98.8  \\ \cline{2-9}
                               & $\eta$   &  1  &  1.0722  ( 0.9847 ) &  0.2271  &  (0.7675, 1.6546)  &  0.03  ( 0.0227 ) &  0.0722  ( -0.0153 ) &  100  \\ \cline{2-9}
                               & $\beta$  &  5  &  4.6928  ( 6.3599 ) &  3.5909  &  (2.0997, 11.8932)  &  1.8951  ( 9.0643 ) &  -0.3072  ( 1.3599 ) &  100  \\ \cline{2-9}
                               & $u$      &  11.51293  &  11.3304  ( 13.8283 ) &  1.2368  &  (7.2405, 12.1245)  &  1.9234  ( 8.5928 ) &  -0.1825  ( 2.3154 ) &  96.4  \\ \cline{2-9}
                               & $\xi$    &  -0.2  &  -0.3548  ( -0.1624 ) &  2.2665  &  (-0.7558, 7.8823)  &  0.1503  ( 0.2999 ) &  -0.1548  ( 0.0376 ) &  100  \\ \cline{2-9}
                               & $\sigma$ &  5  &  5.98  ( 4.4598 ) &  2.3535  &  (0.3139, 8.8821)  &  7.8561  ( 9.3276 ) &  0.98  ( -0.5402 ) &  98.8  \\ \cline{2-9}
                               & $\phi_2$ &  0.1  &  0.1029  ( 0.1122 ) &  0.0198  &  (0.0686, 0.1491)  &  1e-04  ( 0.0014 ) &  0.0029  ( 0.0122 ) &  100  \\ \hline                        
      \multirow{6}{*}{ 300 } & $\phi_1$ &  0.4  &  0.4009  (NA) &  0.0366  &  (0.3419, 0.4892)  &  8e-04  (NA) &  9e-04  (NA) &  99.4  \\ \cline{2-9}
                               & $\eta$   &  1  &  1.0626  ( 0.9744 ) &  0.1968  &  (0.7753, 1.5603)  &  0.0205  ( 0.0162 ) &  0.0626  ( -0.0256 ) &  100  \\ \cline{2-9}
                               & $\beta$  &  5  &  4.623  ( 6.2852 ) &  3.7294  &  (2.3949, 13.4302)  &  1.2301  ( 6.4886 ) &  -0.377  ( 1.2852 ) &  100  \\ \cline{2-9}
                               & $u$      &  11.51293  &  11.3434  ( 13.8601 ) &  1.1666  &  (8.3574, 12.8928)  &  1.2609  ( 8.246 ) &  -0.1695  ( 2.3472 ) &  98.4  \\ \cline{2-9}
                               & $\xi$    &  -0.2  &  -0.2983  ( -0.245 ) &  1.8915  &  (-0.6716, 6.5678)  &  0.0818  ( 0.1843 ) &  -0.0983  ( -0.045 ) &  100  \\ \cline{2-9}
                               & $\sigma$ &  5  &  5.6034  ( 4.7422 ) &  2.2712  &  (0.4516, 8.5295)  &  4.3803  ( 6.5155 ) &  0.6034  ( -0.2578 ) &  99.2  \\ \cline{2-9}
                               & $\phi_2$ &  0.1  &  0.1026  ( 0.1101 ) &  0.0198  &  (0.0669, 0.1478)  &  1e-04  ( 0.0013 ) &  0.0026  ( 0.0101 ) &  100  \\ \hline
                               
   \multirow{6}{*}{ 400 } & $\phi_1$ &  0.4  &  0.401  (NA) &  0.0368  &  (0.3495, 0.4995)  &  6e-04  (NA) &  0.001  (NA) &  99  \\ \cline{2-9}
                               & $\eta$   &  1  &  1.0575  ( 0.9709 ) &  0.2172  &  (0.7863, 1.6433)  &  0.0155  ( 0.0125 ) &  0.0575  ( -0.0291 ) &  100  \\ \cline{2-9}
                               & $\beta$  &  5  &  4.6088  ( 6.1842 ) &  3.6571  &  (2.19, 12.798)  &  1.078  ( 5.006 ) &  -0.3912  ( 1.1842 ) &  100  \\ \cline{2-9}
                               & $u$      &  11.51293  &  11.3791  ( 13.8276 ) &  1.3352  &  (8.5167, 13.8758)  &  0.9787  ( 7.8276 ) &  -0.1338  ( 2.3146 ) &  98  \\ \cline{2-9}
                               & $\xi$    &  -0.2  &  -0.265  ( -0.2397 ) &  2.2304  &  (-0.8132, 7.3489)  &  0.0473  ( 0.1223 ) &  -0.065  ( -0.0397 ) &  100  \\ \cline{2-9}
                               & $\sigma$ &  5  &  5.357  ( 4.6528 ) &  2.9611  &  (0.5012, 10.2852)  &  2.5628  ( 4.7435 ) &  0.357  ( -0.3472 ) &  99.6  \\ \cline{2-9}
                               & $\phi_2$ &  0.1  &  0.1023  ( 0.1107 ) &  0.0216  &  (0.0597, 0.1506)  &  1e-04  ( 0.0012 ) &  0.0023  ( 0.0107 ) &  99.6  \\ \hline
                               
         \multirow{6}{*}{ 500 } & $\phi_1$ &  0.4  &  0.4009  (NA) &  0.0311  &  (0.3492, 0.4744)  &  5e-04  (NA) &  9e-04  (NA) &  99.4  \\ \cline{2-9}
                               & $\eta$   &  1  &  1.0558  ( 0.9683 ) &  0.1638  &  (0.7713, 1.4328)  &  0.0132  ( 0.0108 ) &  0.0558  ( -0.0317 ) &  100  \\ \cline{2-9}
                               & $\beta$  &  5  &  4.5824  ( 6.0976 ) &  2.6737  &  (2.6421, 10.9845)  &  0.9037  ( 3.9509 ) &  -0.4176  ( 1.0976 ) &  100  \\ \cline{2-9}
                               & $u$      &  11.51293  &  11.3942  ( 13.7412 ) &  1.0441  &  (9.1254, 13.3594)  &  0.7906  ( 7.0306 ) &  -0.1188  ( 2.2282 ) &  98.8  \\ \cline{2-9}
                               & $\xi$    &  -0.2  &  -0.2507  ( -0.2357 ) &  1.3818  &  (-0.6334, 5.0342)  &  0.0346  ( 0.0879 ) &  -0.0507  ( -0.0357 ) &  100  \\ \cline{2-9}
                               & $\sigma$ &  5  &  5.278  ( 4.606 ) &  1.97  &  (0.6256, 7.5148)  &  1.9003  ( 3.3344 ) &  0.278  ( -0.394 ) &  99.4  \\ \cline{2-9}
                               & $\phi_2$ &  0.1  &  0.1019  ( 0.1127 ) &  0.0213  &  (0.0608, 0.1481)  &  1e-04  ( 0.0012 ) &  0.0019  ( 0.0127 ) &  100  \\ \hline

    \end{tabular}
\captionsetup{
  labelfont=bf, 
  textfont=normal 
}
    \label{Table-2}
\end{table*}
\end{landscape}

\begin{landscape}
\begin{table*}[htbp]
   \caption{Monte Carlo simulation results ($N=2000$) for sample sizes $n = 150, 200, 300, 400,$ and $500$, with an inlier proportion of 40\% and an exponential-tailed GPD component with a tail fraction of 10\%, reporting sample mean, BSE, 95\% BCI, MSE, bias, and CP.}
    \centering
    \scriptsize 
    \setlength{\tabcolsep}{3pt} 
    \renewcommand{\arraystretch}{1.3} 
    \begin{tabular}{|p{2.5cm}|p{1.8cm}|p{1.8cm}|p{2.4cm}|p{2cm}|p{3cm}|p{2.4cm}|p{2.3cm}|p{1.8cm}|} 
        \hline
         \textbf{$n$}  & \textbf{Parameter} & \textbf{True} & \textbf{Sample Mean} & \textbf{BSE} & \textbf{BCI} & \textbf{MSE} & \textbf{Bias} & \textbf{CP} \\ \hline
      \multirow{6}{*}{ 150 } & $\phi_1$ &  0.4  &  0.4017  (NA) &  0.0504  &  (0.264, 0.4648)  &  0.0017  (NA) &  0.0017  (NA) &  96.6  \\ \cline{2-9}
                               & $\eta$   &  1  &  1.082  ( 0.9913 ) &  0.3873  &  (0.8758, 2.3558)  &  0.0416  ( 0.0326 ) &  0.082  ( -0.0087 ) &  99.6  \\ \cline{2-9}
                               & $\beta$  &  5  &  4.7834  ( 6.6294 ) &  4.2775  &  (1.6728, 14.9815)  &  2.8842  ( 14.4427 ) &  -0.2166  ( 1.6294 ) &  99.6  \\ \cline{2-9}
                               & $u$      &  11.51293  &  11.2564  ( 13.9772 ) &  1.6218  &  (7.5661, 14.1347)  &  2.7279  ( 11.3212 ) &  -0.2565  ( 2.4643 ) &  92.8  \\ \cline{2-9}
                               & $\xi$    &  0  &  -0.1936  ( 0.1315 ) &  4.7132  &  (-1.3287, 12.4424)  &  0.2604  ( 0.6512 ) &  -0.1936  ( 0.1315 ) &  100  \\ \cline{2-9}
                               & $\sigma$ &  5  &  6.342  ( 5.156 ) &  3.7099  &  (0.0118, 11.9709)  &  13.9194  ( 21.2158 ) &  1.342  ( 0.156 ) &  95.6  \\ \cline{2-9}
                               & $\phi_2$ &  0.1  &  0.1036  ( 0.1127 ) &  0.025  &  (0.0595, 0.1651)  &  1e-04  ( 0.002 ) &  0.0036  ( 0.0127 ) &  99  \\ \hline
      \multirow{6}{*}{ 200 } & $\phi_1$ &  0.4  &  0.4014  (NA) &  0.0409  &  (0.327, 0.4889)  &  0.0012  (NA) &  0.0014  (NA) &  98.6  \\ \cline{2-9}
                               & $\eta$   &  1  &  1.066  ( 0.9893 ) &  0.2121  &  (0.7608, 1.581)  &  0.0298  ( 0.0246 ) &  0.066  ( -0.0107 ) &  99.8  \\ \cline{2-9}
                               & $\beta$  &  5  &  4.7664  ( 6.3271 ) &  13.3727  &  (2.3616, 15.1065)  &  2.1343  ( 9.4745 ) &  -0.2336  ( 1.3271 ) &  100  \\ \cline{2-9}
                               & $u$      &  11.51293  &  11.2937  ( 13.8895 ) &  1.1992  &  (7.6729, 12.3516)  &  2.0207  ( 10.6206 ) &  -0.2192  ( 2.3765 ) &  96.6  \\ \cline{2-9}
                               & $\xi$    &  0  &  -0.137  ( 0.0221 ) &  2.237  &  (-0.5681, 8.2515)  &  0.162  ( 0.3924 ) &  -0.137  ( 0.0221 ) &  100  \\ \cline{2-9}
                               & $\sigma$ &  5  &  5.9233  ( 5.3654 ) &  2.4937  &  (0.3038, 9.1679)  &  8.5527  ( 18.1322 ) &  0.9233  ( 0.3654 ) &  98.8  \\ \cline{2-9}
                               & $\phi_2$ &  0.1  &  0.1035  ( 0.1158 ) &  0.0187  &  (0.0709, 0.1476)  &  1e-04  ( 0.0027 ) &  0.0035  ( 0.0158 ) &  99.8  \\ \hline                        
   \multirow{6}{*}{ 300 } & $\phi_1$ &  0.4  &  0.4006  (NA) &  0.0364  &  (0.3525, 0.4994)  &  9e-04  (NA) &  6e-04  (NA) &  99  \\ \cline{2-9}
                               & $\eta$   &  1  &  1.057  ( 0.9796 ) &  0.1856  &  (0.7635, 1.4827)  &  0.0201  ( 0.0169 ) &  0.057  ( -0.0204 ) &  99.8  \\ \cline{2-9}
                               & $\beta$  &  5  &  4.7061  ( 6.1797 ) &  4.8363  &  (2.6206, 15.4328)  &  1.7432  ( 5.9746 ) &  -0.2939  ( 1.1797 ) &  100  \\ \cline{2-9}
                               & $u$      &  11.51293  &  11.3124  ( 13.9208 ) &  1.3129  &  (8.0526, 12.9753)  &  1.3404  ( 10.077 ) &  -0.2005  ( 2.4078 ) &  97.4  \\ \cline{2-9}
                               & $\xi$    &  0  &  -0.0815  ( -0.0426 ) &  2.106  &  (-0.5473, 7.3149)  &  0.0847  ( 0.2424 ) &  -0.0815  ( -0.0426 ) &  100  \\ \cline{2-9}
                               & $\sigma$ &  5  &  5.4862  ( 5.4306 ) &  2.6638  &  (0.4128, 9.2719)  &  4.234  ( 12.437 ) &  0.4862  ( 0.4306 ) &  98.8  \\ \cline{2-9}
                               & $\phi_2$ &  0.1  &  0.1035  ( 0.113 ) &  0.0203  &  (0.064, 0.148)  &  1e-04  ( 0.0019 ) &  0.0035  ( 0.013 ) &  99.8  \\ \hline
                               
       \multirow{6}{*}{ 400 } & $\phi_1$ &  0.4  &  0.4007  (NA) &  0.0338  &  (0.3559, 0.4945)  &  6e-04  (NA) &  7e-04  (NA) &  98.8  \\ \cline{2-9}
                               & $\eta$   &  1  &  1.0528  ( 0.9751 ) &  0.1895  &  (0.7666, 1.5207)  &  0.0158  ( 0.0133 ) &  0.0528  ( -0.0249 ) &  99.8  \\ \cline{2-9}
                               & $\beta$  &  5  &  4.683  ( 6.0891 ) &  4.0013  &  (2.5807, 14.3109)  &  1.2025  ( 4.703 ) &  -0.317  ( 1.0891 ) &  100  \\ \cline{2-9}
                               & $u$      &  11.51293  &  11.3648  ( 13.9104 ) &  1.2972  &  (9.2156, 14.0921)  &  0.991  ( 9.9102 ) &  -0.1481  ( 2.3975 ) &  98  \\ \cline{2-9}
                               & $\xi$    &  0  &  -0.0493  ( -0.044 ) &  2.0165  &  (-0.5505, 7.1442)  &  0.0537  ( 0.1803 ) &  -0.0493  ( -0.044 ) &  100  \\ \cline{2-9}
                               & $\sigma$ &  5  &  5.2682  ( 5.3194 ) &  2.7809  &  (0.5316, 10.2327)  &  2.607  ( 8.9892 ) &  0.2682  ( 0.3194 ) &  99.4  \\ \cline{2-9}
                               & $\phi_2$ &  0.1  &  0.1026  ( 0.1131 ) &  0.021  &  (0.0646, 0.1495)  &  1e-04  ( 0.0018 ) &  0.0026  ( 0.0131 ) &  99.6  \\ \hline
        \multirow{6}{*}{ 500 } & $\phi_1$ &  0.4  &  0.4006  (NA) &  0.0299  &  (0.3382, 0.4579)  &  5e-04  (NA) &  6e-04  (NA) &  99.2  \\ \cline{2-9}
                               & $\eta$   &  1  &  1.0479  ( 0.9732 ) &  0.1603  &  (0.7961, 1.4348)  &  0.0127  ( 0.0115 ) &  0.0479  ( -0.0268 ) &  99.8  \\ \cline{2-9}
                               & $\beta$  &  5  &  4.6823  ( 6.015 ) &  3.1332  &  (2.7569, 11.4238)  &  1.0547  ( 3.8529 ) &  -0.3177  ( 1.015 ) &  100  \\ \cline{2-9}
                               & $u$      &  11.51293  &  11.3799  ( 13.8769 ) &  1.1319  &  (9.4227, 13.556)  &  0.7928  ( 9.6239 ) &  -0.133  ( 2.364 ) &  98.6  \\ \cline{2-9}
                               & $\xi$    &  0  &  -0.0401  ( -0.0375 ) &  1.3478  &  (-0.3975, 4.9998)  &  0.0384  ( 0.1256 ) &  -0.0401  ( -0.0375 ) &  100  \\ \cline{2-9}
                               & $\sigma$ &  5  &  5.2132  ( 5.2268 ) &  1.8801  &  (0.6678, 7.3568)  &  1.9534  ( 6.7681 ) &  0.2132  ( 0.2268 ) &  99.6  \\ \cline{2-9}
                               & $\phi_2$ &  0.1  &  0.1025  ( 0.1144 ) &  0.02  &  (0.0622, 0.1457)  &  1e-04  ( 0.0023 ) &  0.0025  ( 0.0144 ) &  100  \\ \hline

    \end{tabular}
\captionsetup{
  labelfont=bf, 
  textfont=normal 
}
    \label{Table-3}
\end{table*}
\end{landscape}

\begin{table}[!htbp]
\caption{Parameter estimates obtained from the proposed FEVIMM and classical methods (MEP, PSP, and PP) for data with and without inliers.}
\centering
\small
\setlength{\tabcolsep}{4pt} 
\renewcommand{\arraystretch}{1.0}
\begin{tabular}{lcccccccc}
\hline
\textbf{Parameter} 
& \textbf{True}  & \textbf{FEVIMM} 
& \textbf{MEP ($\mathbf{x}$)} & \textbf{MEP ($\mathbf{x}^{+}$)}    
& \textbf{PSP ($\mathbf{x}$)} & \textbf{PSP ($\mathbf{x}^{+}$)} 
& \textbf{PP ($\mathbf{x}$)}  & \textbf{PP ($\mathbf{x}^{+}$)} \\ \hline

\(\phi_1\)  & 0.400   & 0.412     & N/A & N/A & N/A & N/A & N/A & N/A \\
\(\eta\)    & 2.000   & 2.174     & N/A & N/A & N/A & N/A & N/A & N/A \\
\(\beta\)   & 5.000   & 4.313     & N/A & N/A & N/A & N/A & N/A & N/A \\
\(u\)       & 19.448  & 18.174    &   18.000  & 21.000 & 18.000  & 21.000 & 18.000 & 21.000 \\
\(\sigma\)  & 5.000   & 5.492     & 6.500 & 5.200 & 6.500 & 5.200 & N/A & N/A\\
\(\xi\)     & 0.200   & 0.057     & -0.006 & 0.140 & -0.006 & 0.140 & -0.240 & 0.110 \\
\(\phi_2\)  & 0.100   & 0.094     & N/A & N/A & N/A & N/A & N/A & N/A \\ \hline

\end{tabular}
\label{Table-MEP1}
\end{table}

\begin{table}[!htbp]
\caption{Parameter estimates obtained from the proposed FEVIMM and classical methods (MEP, PSP, and PP) for data with and without inliers.}
\centering
\small
\setlength{\tabcolsep}{4pt} 
\renewcommand{\arraystretch}{1.0}
\begin{tabular}{lcccccccc}
\hline
\textbf{Parameter} 
& \textbf{True}  & \textbf{FEVIMM} 
& \textbf{MEP ($\mathbf{x}$)} & \textbf{MEP ($\mathbf{x}^{+}$)}    
& \textbf{PSP ($\mathbf{x}$)} & \textbf{PSP ($\mathbf{x}^{+}$)} 
& \textbf{PP ($\mathbf{x}$)}  & \textbf{PP ($\mathbf{x}^{+}$)} \\ \hline

\(\phi_1\)  & 0.200   & 0.178     & N/A & N/A & N/A & N/A & N/A & N/A \\
\(\eta\)    & 1.000   & 1.036     & N/A & N/A & N/A & N/A & N/A & N/A \\
\(\beta\)   & 5.000   & 4.526     & N/A & N/A & N/A & N/A & N/A & N/A \\
\(u\)       & 11.513  & 11.487    & 11.000  & 12.000 & 11.000  & 12.000 & 11.000 & 12.000 \\
\(\sigma\)  & 5.000   & 5.129     & 5.900 & 6.000 & 5.900 & 6.000 & N/A & N/A\\
\(\xi\)     & 0.200   & 0.157     & 0.055 & 0.0051 & 0.055 & 0.0051 & 0.084 & 0.210 \\
\(\phi_2\)  & 0.100   & 0.104     & N/A & N/A & N/A & N/A & N/A & N/A \\ \hline

\end{tabular}
\label{Table-MEP2}
\end{table}

\begin{table}[!htbp]
\caption{Parameter estimates obtained from the proposed FEVIMM and classical methods (MEP, PSP, and PP) for data with and without inliers.}
\centering
\small
\setlength{\tabcolsep}{4pt} 
\renewcommand{\arraystretch}{1.0}
\begin{tabular}{lcccccccc}
\hline
\textbf{Parameter} 
& \textbf{True}  & \textbf{FEVIMM} 
& \textbf{MEP ($\mathbf{x}$)} & \textbf{MEP ($\mathbf{x}^{+}$)}    
& \textbf{PSP ($\mathbf{x}$)} & \textbf{PSP ($\mathbf{x}^{+}$)} 
& \textbf{PP ($\mathbf{x}$)}  & \textbf{PP ($\mathbf{x}^{+}$)} \\ \hline

\(\phi_1\)  & 0.400   & 0.403     & N/A & N/A & N/A & N/A & N/A & N/A \\
\(\eta\)    & 1.000   & 0.999     & N/A & N/A & N/A & N/A & N/A & N/A \\
\(\beta\)   & 5.000   & 6.312     & N/A & N/A & N/A & N/A & N/A & N/A \\
\(u\)       & 8.0472  & 10.855    &   11.000  & 15.000 & 11.000  & 15.000 & 11.000 & 15.000 \\
\(\sigma\)  & 5.000   & 6.823    & 7.100 & 9.000 & 7.100 & 9.000 & N/A & N/A\\
\(\xi\)     & 0.200   &  0.113     & 0.086 & -0.048 & 0.086 & -0.048 & -0.240 & -0.250 \\
\(\phi_2\)  & 0.200   & 0.109     & N/A & N/A & N/A & N/A & N/A & N/A \\ \hline
\end{tabular}
\label{Table-MEP3}
\end{table}

\subsection{Comparative Analysis of Bias and MSE: FEVIMM, EVMM, and FEVMM}\label{smstudyc}
\noindent
To assess how well the proposed FEVIMM performs at parameter estimation compared with the existing FEVMM and EVMM models, we considered various scenarios. In each case, the data were generated from FEVIMM. The performance of the estimators was evaluated using bias and MSE across sample sizes \( n \in \{150, 200, 300, 400, 500, 750, 1000\} \), based on Monte Carlo replications (\(N = 2000\)). The bias and MSE were computed for the parameter estimates obtained under the EVMM, FEVMM, and FEVIMM. 
The selected simulation settings are $\phi_1=0.4$, $\eta=4$, $\beta=1$, $u=6.6807$, $\xi=0.2$, $\sigma=4$, and $\phi_2=0.10$.
\noindent Figure~\ref{fig: bias comparison1} shows that the FEVIMM consistently yields lower bias and MSE across most parameter estimates than both the FEVMM and EVMM. The improvements are especially clear for the threshold ($u$), scale ($\sigma$), shape ($\xi$), and tail fraction ($\phi_2$) (the proportion of extremes), indicating that FEVIMM performs better at recovering these parameters under this data-generating process.
Across the scenarios, the observed trends remain consistent, showing similar patterns in the estimators' performance.

\begin{figure}[!htbp]
		\centering
		\includegraphics[width=18.5cm,height=20.0cm]{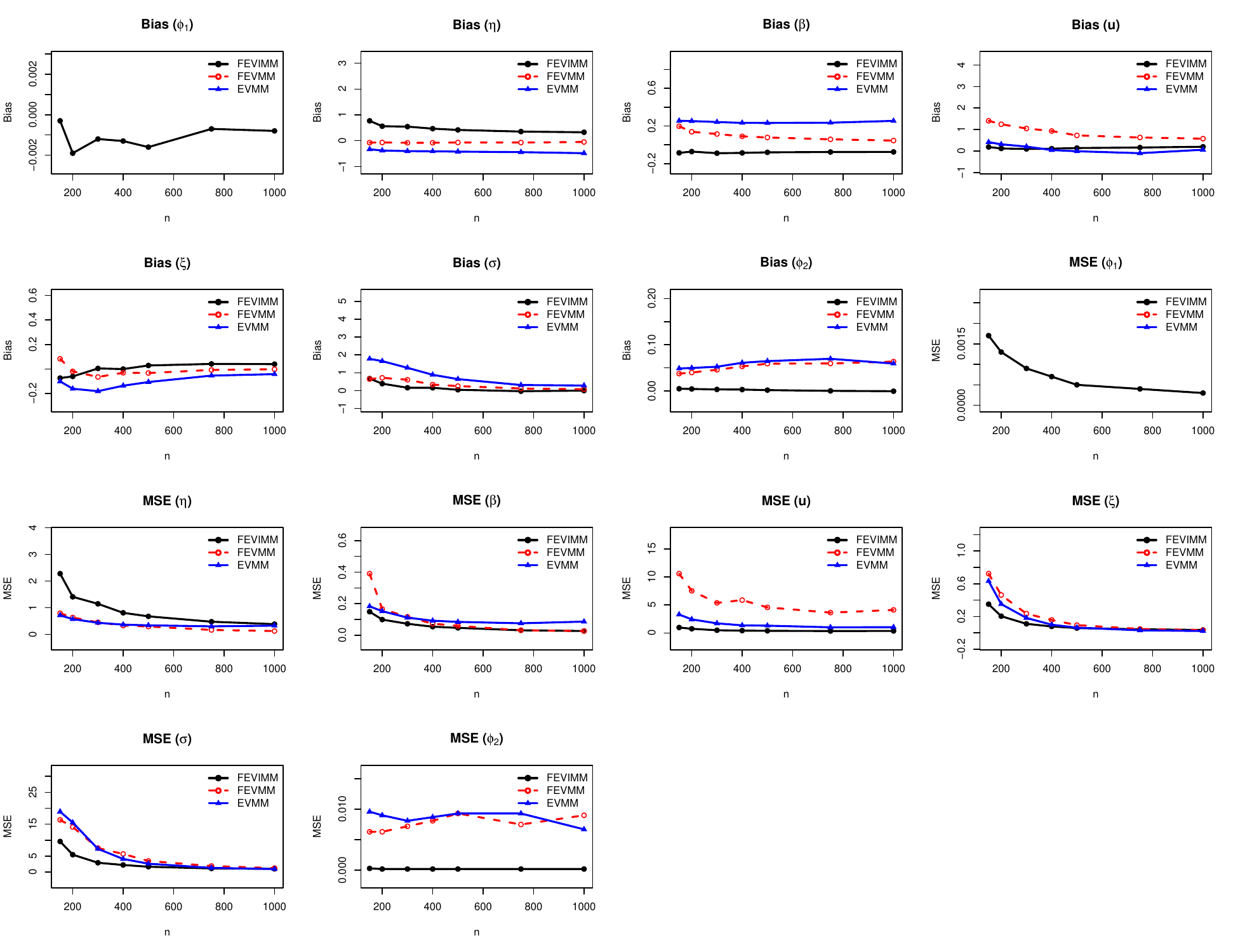}
        \caption{Bias and MSE comparison plots under Scenario 1, based on Monte Carlo replications ($N = 2000$): \(\phi_1 = 0.4\), \(\eta = 4\), \(\beta = 1\), \(u = 6.6807\), \(\xi = 0.2\), \(\sigma = 4\), \(\phi_2 = 0.10\).}
		\label{fig: bias comparison1}
   \end{figure} 

\subsection{Sensitivity Analysis with Respect to Inlier Proportions}
\noindent 
 To assess the robustness of the proposed FEVIMM under varying inlier proportions, a sensitivity analysis was conducted by simulating data from the FEVIMM~\eqref{evmmlpdf1} for
$\phi_1 \in \{0.1,0.2,0.3,0.4,0.5\}$. The remaining parameters were set to
$\eta=1$, $\beta=5$, $\sigma=5$, $\xi=0.2$, $\phi_2=0.1$, and $u=11.51$.
The analysis considered $n \in \{500,750,1000\}$ with $N=2000$ replications for each case.
Parameters were estimated using the MLE procedure described in Section~\ref{estimation}, and the corresponding bias and MSE were computed.

Figures~\ref{fig:inlier_sensitivity1} presents the bias and MSE for $n=1000$.
The bias remains close to zero across the considered inlier proportions, while the MSE is generally stable. For $\eta$ and $\beta$, the MSE increases slightly with $\phi_1$, reflecting the reduced number of non-zero observations available for estimating the bulk distribution parameters. Despite these variations, the proposed FEVIMM consistently yields lower bias and MSE than the FEVMM and EVMM, which do not account for inliers, demonstrating its robustness to varying inlier proportions.

\vspace{-0.2cm}
\begin{figure}[!htbp]
\centering
\makebox[\textwidth][c]{%
    \includegraphics[width=1.18\textwidth]{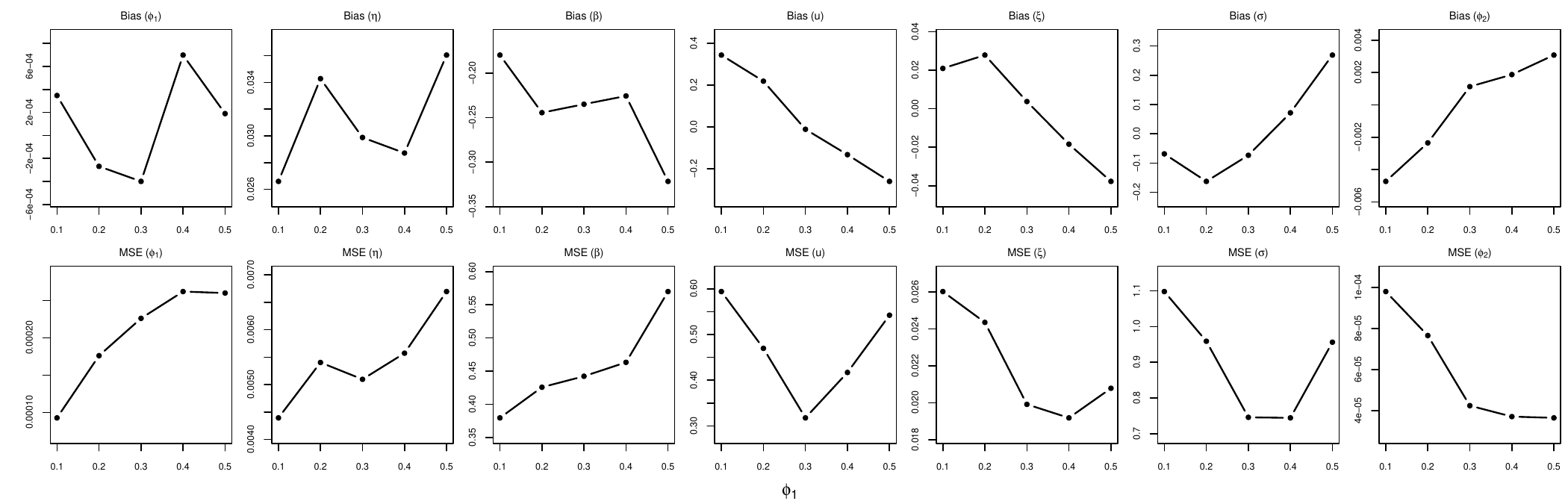}%
}
\vspace{-0.2cm}
\caption{Monte Carlo bias and MSE of the FEVIMM parameter estimates for varying inlier proportion $\phi_1$.}
\label{fig:inlier_sensitivity1}
\end{figure}

\vspace{-0.1cm}
\section{Application to Real Data}\label{arStudy2}
\noindent In this section, we illustrate the proposed methodology using two real-world datasets: weekly syphilis case counts from the \texttt{ZIM} package in \texttt{R} and monthly rainfall data for Pune. We compare the estimates of $u$ and extreme parameters such as $\xi$, $\sigma$, and $\phi_2$ with those obtained from classical methods including MEP (\cite{coles2001introduction}), PSP (\cite{coles2001introduction}), and PP (\cite{finkenstadt2003extreme}), as well as mixture models such as EVMM (\cite{behrens2004bayesian}) and FEVMM (\cite{macdonald2011flexible}).
We study risk measures for both datasets and present the corresponding results for the FEVIMM, EVMM, and FEVMM (see Figure~\ref{fig:return_syphilis} and \ref{fig:return_rainfall}). The MEP, the PSP, and PP are plotted for both real datasets (see Figure~\ref{fig:meanexcessreal}). 
\vspace{-0.2cm}
\subsubsection*{Illustration 1: Weekly Syphilis Cases in Maryland}
\noindent This study considers the real-life dataset \texttt{syph} from the \texttt{ZIM} package in \texttt{R}. The dataset records the weekly number of syphilis cases in the United States from $2007$ to $2010$, containing $209$ observations across $69$ variables. For analysis, we focus on the variable \texttt{a33} (Maryland), which represents the weekly reported cases in Maryland, and present the estimated results in Table~\ref{Table-21}.
\vspace{-0.1cm}\begin{table}[!htbp]
\caption{Parameter estimates from different models with standard errors in parentheses.}
\centering
\scriptsize
\setlength{\tabcolsep}{6pt} 
\renewcommand{\arraystretch}{1.2} 
\captionsetup{
  labelfont=bf,
  textfont=normal
}
\begin{tabular}{lcccc}
\hline
\textbf{Parameter} & \textbf{EVMM} & \textbf{FEVMM} & \textbf{FEVIMM} & \textbf{\(p\)-value} \\
\hline
\(\phi_1\) & N/A & N/A & 0.2993 (0.0064) & \(< 0.001\) \\
\(\eta\)   & 3.7702  (0.0104) & 2.1087 (0.05202) & 3.1322 (0.0064) & \(< 0.001\) \\
\(\beta\)  & 1.2562 (0.0043) & 2.2791 (0.1354) & 1.8197 (0.0068) & \(< 0.001\) \\
\(u\)      & 7.0008 (0.0244) &  4.9994 (0.0225)  & 7.0000 (0.0002) & \(< 0.001\) \\
\(\xi\)    & -0.3997 (0.0069) & 0.0859 (0.0094)  & -0.3354 (0.0064) & \(< 0.001\) \\
\(\sigma\) & 3.7946 (0.0291) & 1.6326 (0.0246)   & 3.6849 (0.0064) & \(< 0.001\) \\
\(\phi_2\) & N/A & 0.5067 (0.04082) & 0.0823 (0.0080) & \(< 0.001\) \\
\hline
\end{tabular}
\label{Table-21}
\end{table}
\vspace{-0.2cm}

\subsubsection*{Illustration 2: Monthly Rainfall in Pune, India}
\noindent
This study considers a real-world weather dataset from Pune, Maharashtra, India, collected from $1965$ to $2002$. The rainfall data were obtained from Kaggle: \href{https://www.kaggle.com/datasets/abhishekmamidi/precipitation-data-of-pune-from-1965-to-2002}{Pune Precipitation Dataset}. The dataset contains $456$ monthly precipitation records (in millimetres). Pune’s population is growing fast, and it is becoming an important financial and business centre. Extreme rainfall can cause financial losses and disrupt business activities, while drought and water availability can also affect the city’s economy and daily life. The estimated parameters based on the Pune rainfall dataset are given in Table~\ref{Table-22}.
\begin{table}[!htbp]
\caption{Parameter estimates from different models with standard errors in parentheses.}
\centering
\scriptsize
\setlength{\tabcolsep}{6pt} 
\renewcommand{\arraystretch}{1.2} 
\captionsetup{
  labelfont=bf,
  textfont=normal
}
\begin{tabular}{lcccc}
\hline
\textbf{Parameter} & \textbf{EVMM} & \textbf{FEVMM} & \textbf{FEVIMM} & \textbf{\(p\)-value} \\
\hline
\(\phi_1\) & N/A & N/A & 0.1134 (0.0148) & \(< 0.001\) \\
\(\eta\)   & 0.3338 (0.0190) & 0.3344 (0.0193) & 0.3382 (0.0175) & \(< 0.001\) \\
\(\beta\)  & 433.3581 (48.4452) & 444.4409 (91.9913) & 371.7527 (22.323) & \(< 0.001\) \\
\(u\)      & 422.9121 (0.0036) & 405.3587 (0.0044) & 402.5191 (0.0021) & \(< 0.001\) \\
\(\xi\)    & -0.4881 (0.2612) & -0.4855 (0.2304) & -0.2593 (0.1403) & \(< 0.005\) \\
\(\sigma\) & 201.6985 (62.0240) & 208.7014 (58.5758) & 156.8195 (28.5226) & \(< 0.001\) \\
\(\phi_2\) &  N/A & 0.1015 (0.0115) & 0.1012 (0.0141) & \(< 0.001\) \\
\hline
\end{tabular}
\label{Table-22}
\end{table}

\begin{figure}[!htbp]
    \centering
    \includegraphics[width=1.1\textwidth]{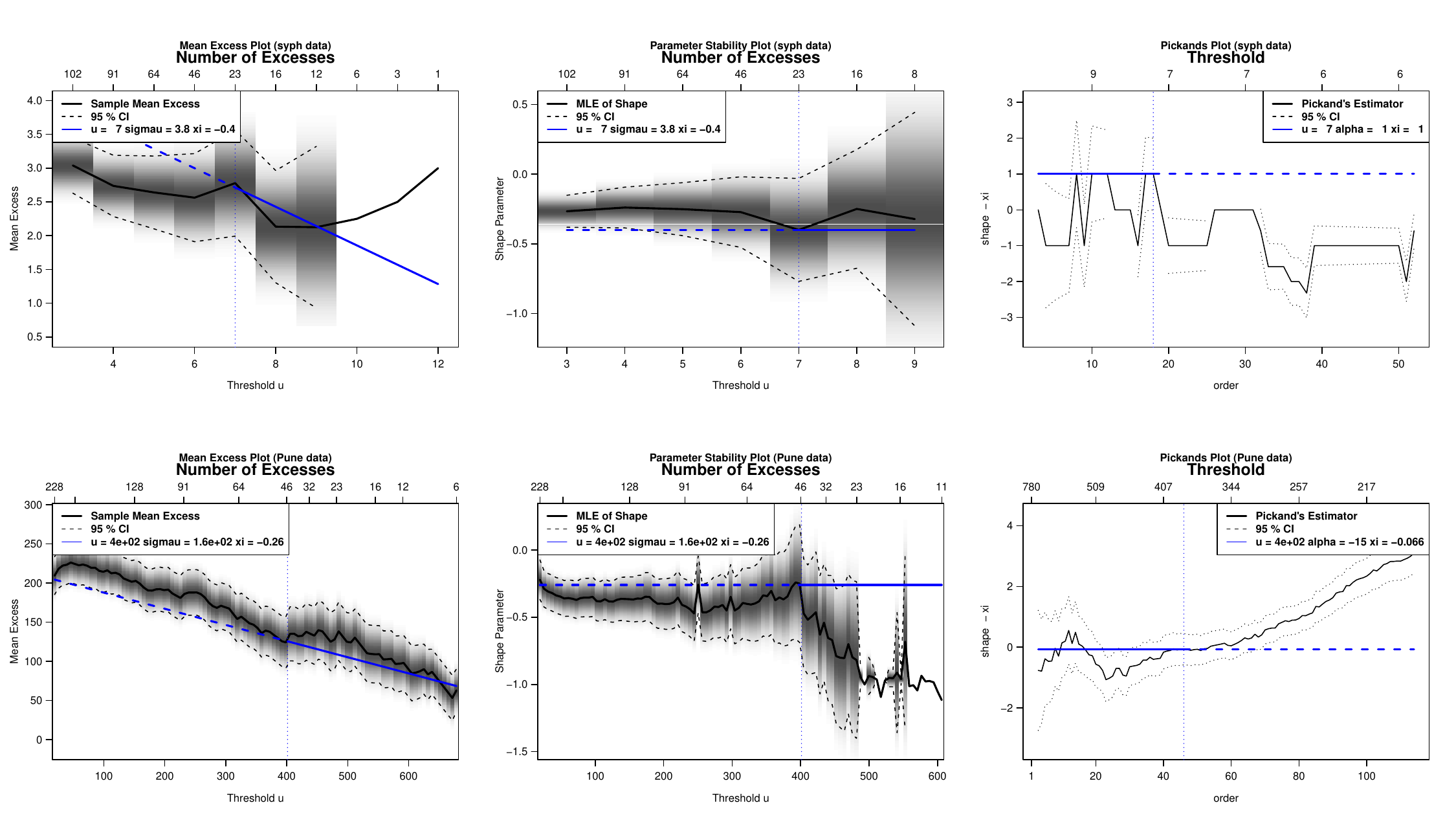}
    \caption{MEP, PSP, and PP for the syphilis and Pune rainfall datasets.}
    \label{fig:meanexcessreal}
\end{figure}
\noindent \textbf{Model Assessment and Comparative Analysis:} 
To assess the performance of the proposed FEVIMM against existing models (FEVMM and EVMM), we consider the normalized Akaike Information Criterion (nAIC) \cite{cohen2021normalized} and standard goodness-of-fit test, namely the $AD$, $CvM$, and $KS$ statistics. For each model (FEVIMM, FEVMM, and EVMM), the null hypothesis ($H_0$) states that the observed data follow the fitted distribution, i.e., $H_0: F(x) = F_{\hat{\theta}_M}(x)$, where $M$ denotes the model and $F_{\hat{\theta}_M}(x)$ its CDF with estimated parameters. The corresponding $p$-values are obtained using a parametric bootstrap (3{,}000 replications). The test statistics, $p$-values, and nAIC are reported in Table~\ref{tab:gof_results}. For each test, models with lower test statistics and higher $p$-values provide a better fit. The estimated parameters are presented in Tables~\ref{Table-21} and \ref{Table-22} for the syphilis and rainfall datasets. The results show that FEVIMM achieves the lowest nAIC and provides the best fit across all three goodness-of-fit tests among the models considered in both datasets. Furthermore, graphical assessments are provided in Figure~\ref{fig:return_comparison}. Figures~\ref{fig:hist_syphilis} and \ref{fig:hist_rainfall} show histograms of the syphilis and rainfall datasets, overlaid with fitted densities, highlighting the agreement between the data and the model. Figures~\ref{fig:return_syphilis} and \ref{fig:return_rainfall} present return level plots against return period, along with 95\% confidence intervals, illustrating the reliability of FEVIMM in estimating extreme quantiles relevant to real-world risk assessment. These findings demonstrate the ability of FEVIMM to capture key distributional features, including extreme events, the threshold, tail fraction, inliers, and the bulk of the distribution.
\begin{table}[!htbp]
\caption{Goodness-of-fit statistics for the Syphilis and Pune Rainfall datasets.}
\centering
\label{tab:gof_results}
\setlength{\tabcolsep}{3.2pt}
\renewcommand{\arraystretch}{1.1}
\begin{tabular}{|c|c|c|c|c|}
\hline
\textbf{Model} & $AD$ ($p$-value) & $CvM$ ($p$-value) & $KS$ ($p$-value) & \textbf{nAIC} \\
\hline
\multicolumn{5}{|c|}{\textbf{Syphilis Dataset}} \\
\hline
EVMM  & 344.89 ($\scriptstyle <0.0001$) & 2.4888 ($\scriptstyle <0.0001$) & 0.3168 ($\scriptstyle <0.0001$) & 4.5684 \\
FEVMM & 364.42 ($\scriptstyle <0.0001$) & 6.0075 ($\scriptstyle <0.0001$) & 0.2825 ($\scriptstyle <0.0001$) & 4.4953 \\
FEVIMM & 17.10 ($\scriptstyle 0.7848$) & 2.9253 ($\scriptstyle 0.8056$) & 0.1735 ($\scriptstyle >0.999$) & 4.4486 \\
\hline
\multicolumn{5}{|c|}{\textbf{Pune Rainfall Dataset}} \\
\hline
EVMM  & 5215.62 ($\scriptstyle <0.0001$) & 56.374 ($\scriptstyle <0.0001$) & 0.1426 ($\scriptstyle <0.0001$) & 10.4731 \\
FEVMM & 4657.30 ($\scriptstyle <0.0001$) & 54.8975 ($\scriptstyle <0.0001$) & 0.1432 ($\scriptstyle <0.0001$) & 10.4767 \\
FEVIMM & 6.9252 ($\scriptstyle 0.4023$) & 0.3607 ($\scriptstyle 0.9997$) & 0.0822 ($\scriptstyle >0.999$) & 9.9911 \\
\hline
\end{tabular}
\label{tab:gof_results}
\end{table}

\begin{figure}[!htbp]
    \centering
    \begin{subfigure}{0.48\textwidth}
        \centering
        \includegraphics[width=7.5cm, height=7cm]{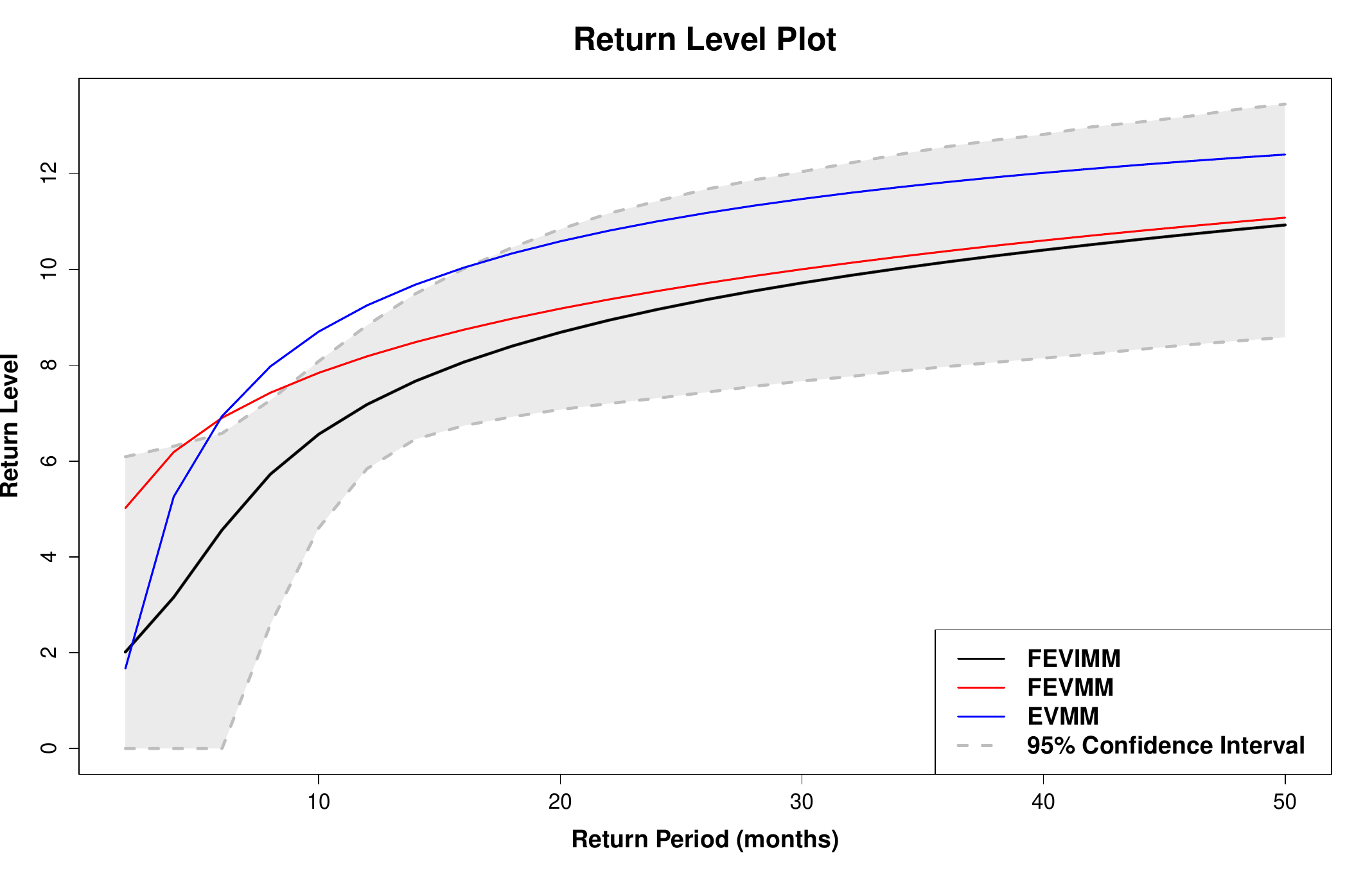}
        \caption{Return level estimates: FEVIMM (with 95\% confidence intervals), FEVMM, and EVMM (Syphilis dataset).}
        \label{fig:return_syphilis}
    \end{subfigure}
    \hfill
    \begin{subfigure}{0.48\textwidth}
        \centering
        \includegraphics[width=7.5cm, height=7cm]{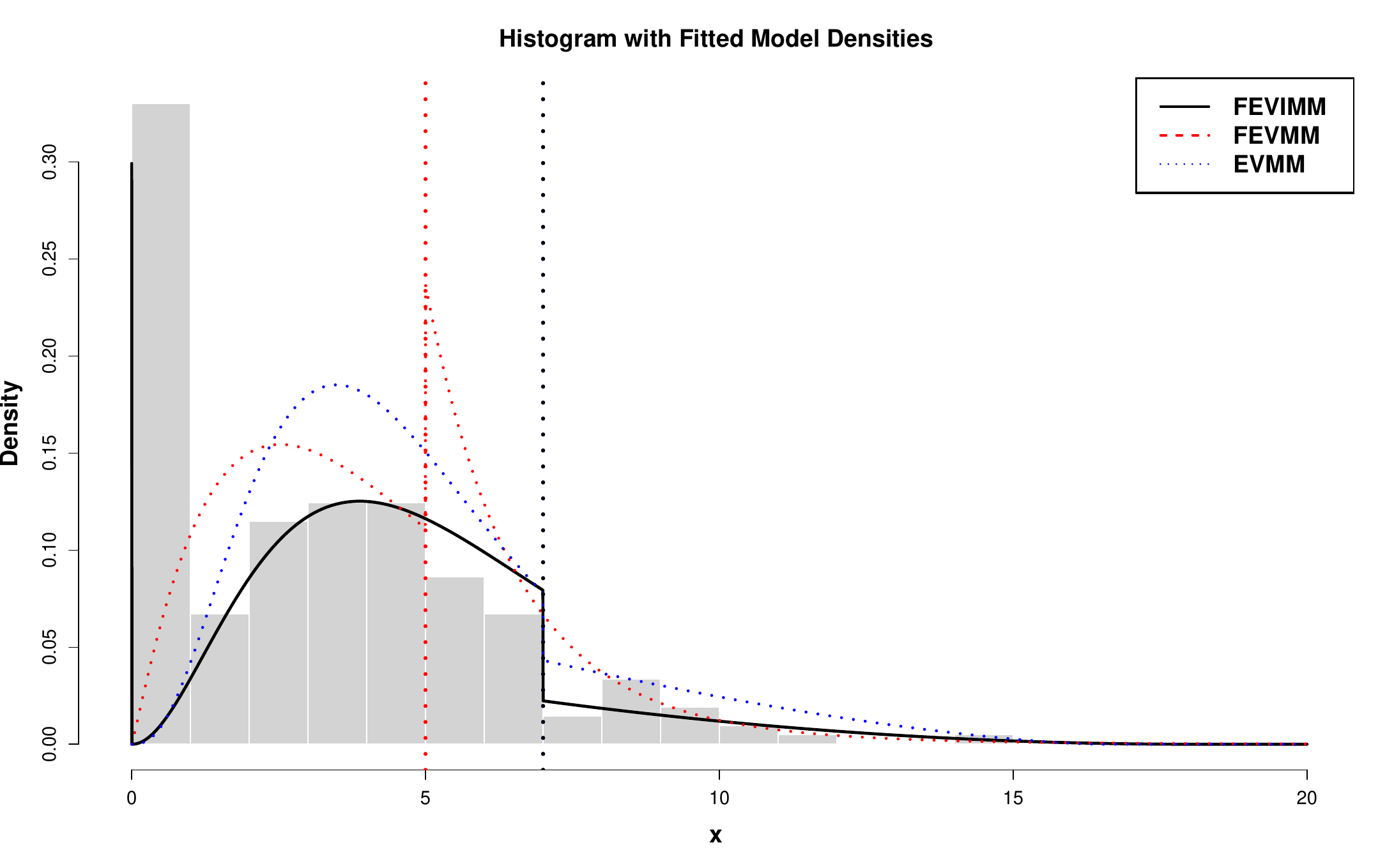}
        \caption{Histogram with fitted model densities. Vertical dashed lines pass through the model-estimated thresholds (Syphilis dataset).}
        \label{fig:hist_syphilis}
    \end{subfigure}
    \\
    \begin{subfigure}{0.48\textwidth}
        \centering
        \includegraphics[width=7.5cm, height=7cm]{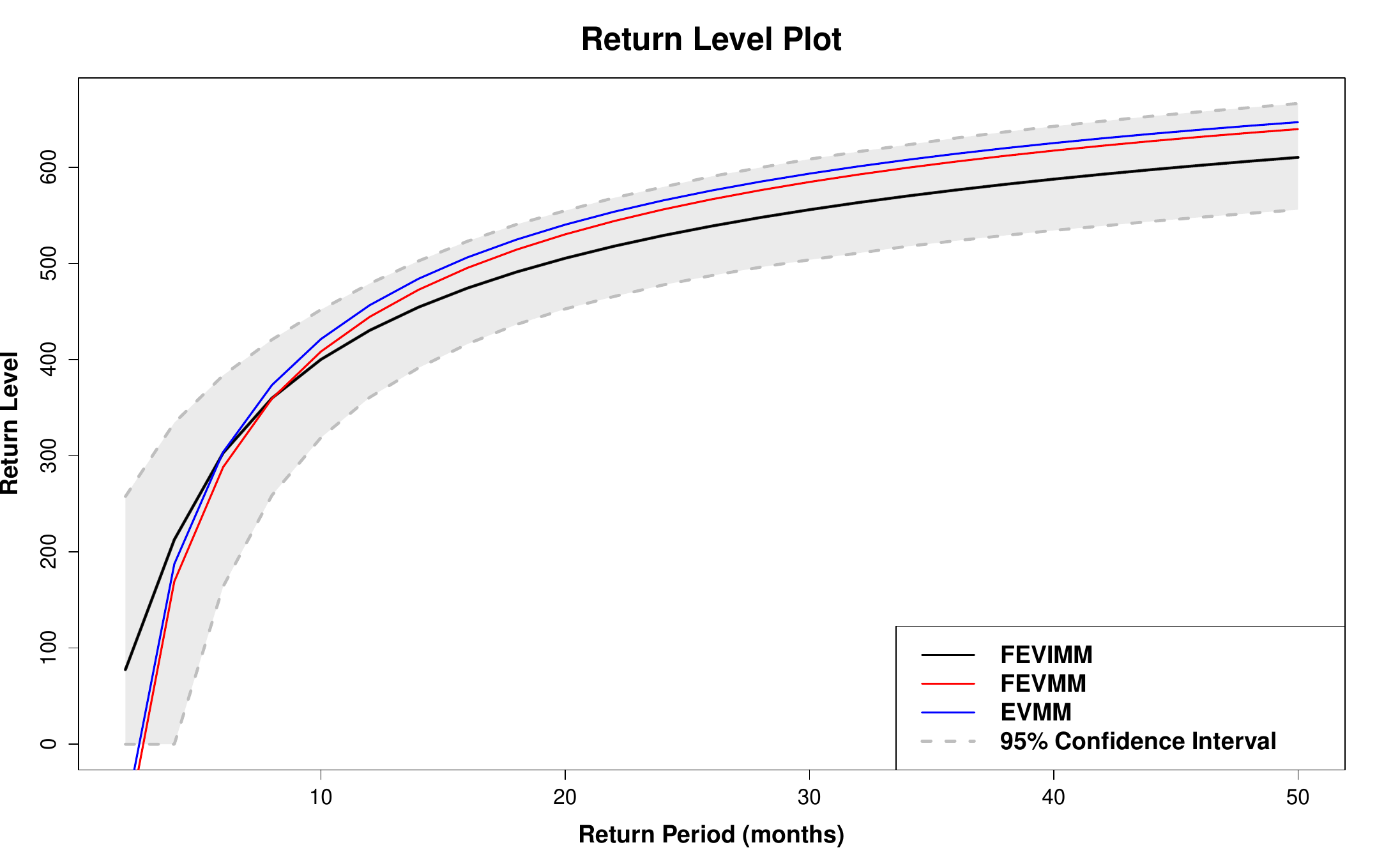}
        \caption{Return level estimates: FEVIMM (with 95\% confidence intervals), FEVMM, and EVMM (Pune rainfall dataset).}
        \label{fig:return_rainfall}
    \end{subfigure}
    \hfill
    \begin{subfigure}{0.48\textwidth}
        \centering
        \includegraphics[width=7.5cm, height=7cm]{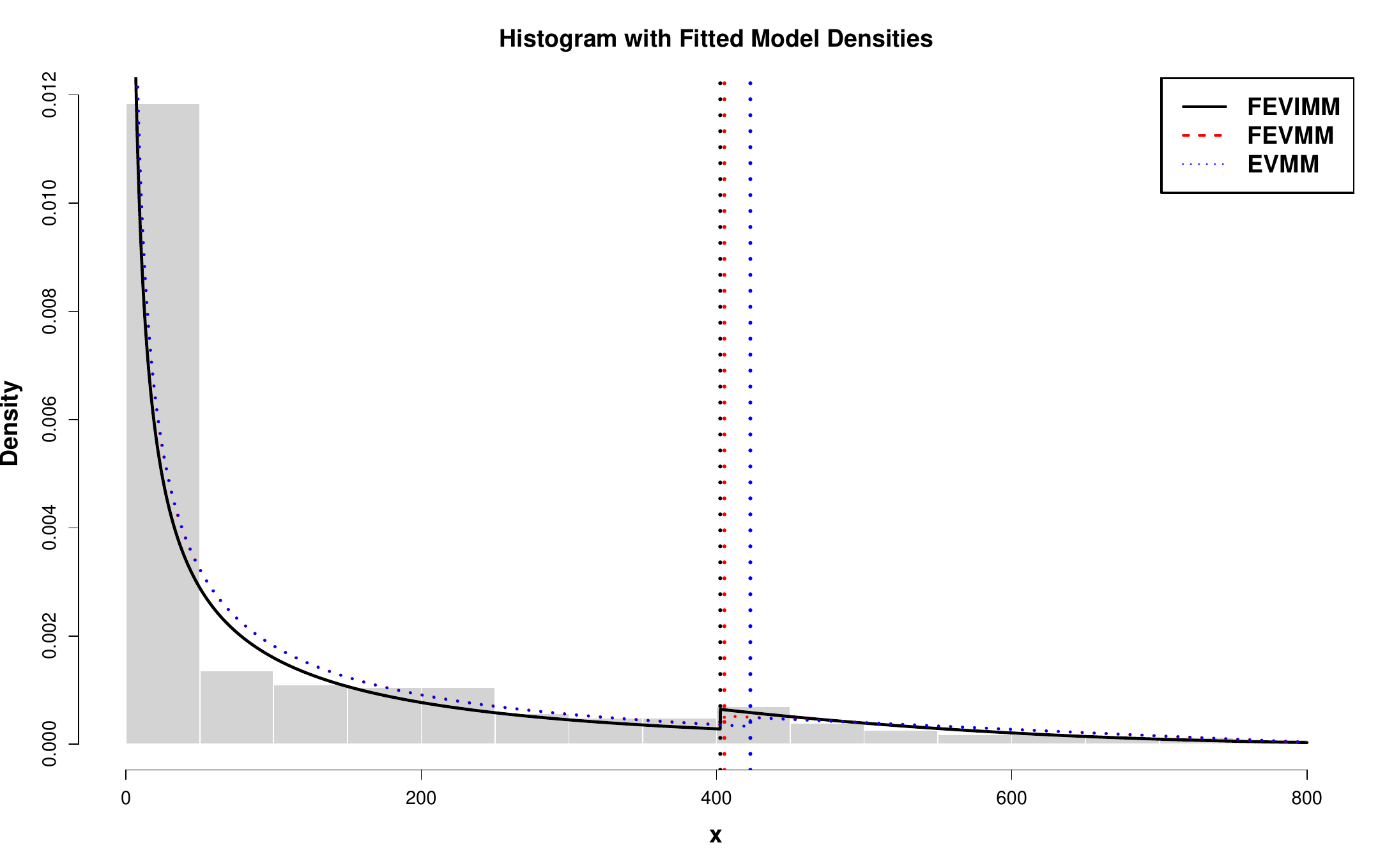}
        \caption{Histogram with fitted model densities. Vertical dashed lines pass through the model-estimated thresholds (Pune rainfall dataset).}
        \label{fig:hist_rainfall}
    \end{subfigure}
    \caption{Return level plots and density fits for real datasets.}
    \label{fig:return_comparison}
\end{figure}

\section{Conclusion and Possible Extensions} \label{conclusion}
\noindent In this article, we proposed a novel framework, FEVIMM, that effectively captures extreme events, including their proportions as parameters, as well as inliers. This is important because the form chosen below the threshold can affect the threshold and extreme parameter estimates used to quantify future extreme events.
In general, the presence of inliers is often assumed not to affect the estimation of extreme parameters, such as the threshold or tail fraction. However, our simulation settings and real-data applications show that ignoring inliers can affect these estimates. 
The theoretical results further show that the mixture proportions $\phi_1$ (inlier proportion) and $\phi_2$ (extremes proportion) are dependent under the unit-sum constraint, resulting in a negative asymptotic correlation between their MLEs.
Accounting for this dependence improves inference for the extreme component and its associated risk measures. The simulation study shows that the proposed FEVIMM yields lower bias and MSE than FEVMM and EVMM for the extreme parameters, as shown in Figure~\ref{fig: bias comparison1} and Table~\ref{Table-1}--\ref{Table-3}.
The FEVIMM provides the best fit to the data, as evidenced by the goodness-of-fit test results and nAIC (see Table~\ref{tab:gof_results}). Figures~\ref{fig:hist_syphilis} and \ref{fig:hist_rainfall} show the real-data histograms with fitted density curves from EVMM, FEVMM, and FEVIMM, highlighting the agreement of FEVIMM with the data. The proposed framework provides a more accurate representation of the data in such scenarios, leading to more precise return-level estimates and better predictions of extreme events.

The proposed model focuses solely on nonnegative values. To the best of our knowledge, research on modeling extremes and inliers is limited.
It is interesting to examine whether inliers occurring at a single or multiple points, or arising from multimodal data structures, affect the threshold position and, consequently, the estimation of GPD parameters.
Natural extensions include accommodating inliers clustering around a point and considering degeneracy at multiple discrete points.
Another research direction is to model multimodality in the bulk distribution using mixture distributions. The advantage of this model formulation is flexibility; the information criteria can be used to determine the number of components.
Finally, extending the FEVIMM to a change-point setup represents another promising research avenue. The change-point models, which explicitly study distributional changes in the structure of extreme values, are extremely limited in the literature \cite{hazra2025estimating}. Moreover, this type of framework is needed when the data contains inliers or when the unimodal assumption does not hold; studying extremes by developing models to handle such cases could provide further insight. 
In summary, this area remains largely underexplored, highlighting a substantial opportunity for further research. 
Future work could explore Bayesian estimation and alternative bulk distributions in the presence of inliers to obtain more robust solutions. The model may also be evaluated further in applications involving reliability, environmental data, finance, and medicine. 
\section*{Acknowledgments}
\noindent The authors would like to thank the Editor and the Reviewers for their helpful comments and suggestions, which
improved the flow and clarity of this manuscript. 
A preprint is available on arXiv; see \cite{nila2026flexible}.
 \bibliographystyle{apalike}
 \bibliography{01ref}
\section*{Appendix A1. Proofs of Propositions}
\label{Appendix_Proof_prepositions}
\subsection*{Proposition 1}
\begin{proof}
From \eqref{evmmlpdf1}, \eqref{eq:gpd_density}, evaluating the left and right limits at $u$ yields
\[
f(u^-) = (1 - (\phi_1 + \phi_2)) \frac{g^*(u \mid \Phi)}{G^*(u \mid \Phi)}, 
\quad
f(u^+) = \phi_2 \, g(u \mid u, \sigma, \xi) = \frac{\phi_2}{\sigma},
\]
where $G^*(\cdot \mid \Phi)$ and $g^*(\cdot \mid \Phi)$ denote the CDF and DF of the bulk distribution, respectively. 
Imposing continuity at $u$, i.e., $f(u^-) = f(u^+)$, yields
\[
(1 - (\phi_1 + \phi_2)) \frac{g^*(u \mid \Phi)}{G^*(u \mid \Phi)} = \frac{\phi_2}{\sigma}.
\]
Solving for $\sigma$, we obtain
\[
\sigma
=
\frac{\phi_2}{1 - \phi_1 - \phi_2}
\;
\frac{G^*(u \mid \Phi)}{g^*(u \mid \Phi)}.
\]
This completes the proof.
\end{proof}
\subsection*{Proposition 2}
\begin{proof}
From \eqref{evmmlpdf1}, evaluating the derivatives of the density from the left and right at $u$ yields
\[
f'(u^-) = (1 - (\phi_1 + \phi_2)) \frac{g^{*\,\prime}(u \mid \Phi)}{G^*(u \mid \Phi)}, ~\text{and}~\quad
f'(u^+) = \phi_2 \, g'(u \mid u, \sigma, \xi),
\]
here $g^{*\,\prime}(\cdot \mid \Phi)$ and $g'(\cdot \mid u,\sigma,\xi)$ denotes derivatives with respect to $x$.
So from \eqref{eq:gpd_density},
\[
g'(u \mid u, \sigma, \xi) = -\frac{1+\xi}{\sigma^2}
\;\Longrightarrow\;
f'(u^+) = -\frac{\phi_2 (1+\xi)}{\sigma^2},
\]
imposing differentiability at $u$, i.e., $f'(u^-) = f'(u^+)$, gives
\begin{equation}
\label{eq:differentiability_step}
(1 - (\phi_1 + \phi_2)) \frac{g^{*\,\prime}(u \mid \Phi)}{G^*(u \mid \Phi)}
=
-\frac{\phi_2 (1+\xi)}{\sigma^2}.
\end{equation}
When the $G^*(\cdot \mid \Phi)$ gamma distribution has shape $\eta$ and scale $\beta$, after solving, we have,
\begin{equation}
\label{eq:contiunitey_step}
g^{*\,\prime}(u \mid \eta,\beta)
=
g^*(u \mid \eta,\beta)\left( \frac{\eta - 1}{u} - \frac{1}{\beta} \right), 
\end{equation}
and from the Proposition~\ref{rem:Continuity} we have, 
\begin{equation}
\label{eq:continutiy_gamma}
(1 - (\phi_1 + \phi_2))\,\frac{g^*(u \mid \eta,\beta)}{G^*(u \mid \eta,\beta)}
=
\frac{\phi_2}{\sigma}
\end{equation}
substituting \eqref{eq:contiunitey_step} and \eqref{eq:continutiy_gamma} into \eqref{eq:differentiability_step}, and solving for $\sigma$, we obtain
\[
\sigma
=
\frac{1+\xi}{
\dfrac{1}{\beta} - \dfrac{\eta-1}{u}},
\]
provided $\sigma > 0$ and $\dfrac{1}{\beta} - \dfrac{\eta-1}{u} \neq 0$.

\noindent
This completes the proof.
\end{proof}
\section*{Appendix A2. Score Functions for the MLE}
\label{Appendix_asymptoticdistribution}
\noindent Recall that the DF $f(x \mid \phi_1, \eta, \beta, u, \xi, \sigma, \phi_2)$ is given below in~\eqref{evmmlpdf}. Here, $g(x \mid u, \sigma, \xi)$ denotes the GPD DF defined in~\eqref{eq:gpd_density}.
\begin{equation}\label{evmmlpdf}
\bm{f}(x \mid \phi_1, \eta, \beta, u, \xi, \sigma,\phi_2) =
\begin{cases} 
\phi_1, & \text{if } x = 0, \\[4pt]
\left( 1 - \phi_1 - \phi_2 \right)
\displaystyle\frac{ \frac{1}{\Gamma(\eta)\beta^\eta} x^{\eta-1} \exp\left(-\frac{x}{\beta}\right) }{ \int_0^u \frac{1}{\Gamma(\eta)\beta^\eta} t^{\eta-1} \exp\left(-\frac{t}{\beta}\right) dt }, & \text{if } 0 < x < u, \\[6pt]
\phi_2 \cdot g(x \mid u, \sigma, \xi), & \text{if } x \geq u,
\end{cases}
\end{equation} the parameters satisfy $0 < \phi_1, \phi_2 < 1$, $\eta > 0$, $\beta > 0$, $u > 0$, $\sigma > 0$, and $\xi \in \mathbb{R}$.
\section*{Score Functions: Case \(\xi \ne 0\)}
\noindent Let \(\boldsymbol{\theta} = (\phi_1, \eta, \beta, \xi, \sigma, \phi_2)\) be the parameter vector and, for a single observation, the log-likelihood function (\ref{eq:log_likelihood}) is given by
\begin{equation}\label{eq:loglikelihood_single}
\ell(\boldsymbol{\theta}; x) =
\begin{cases}
\log(\phi_1), & \text{if } x = 0, \\[10pt]

\log(1 - \phi_1 - \phi_2)
- \log G^*(u_{0} \mid \eta, \beta)
+ (\eta - 1)\log x
- \dfrac{x}{\beta} \\[4pt]
\qquad\quad
- \eta \log \beta
- \log \Gamma(\eta),
& \text{if } 0 < x < u_{0}, \\[12pt]

\log(\phi_2)
- \log \sigma
- \left( \dfrac{1}{\xi} + 1 \right)
\log\left( 1 + \xi \dfrac{x - u_{0}}{\sigma} \right),
& \text{if } x \geq u_{0}.
\end{cases}
\end{equation}
\noindent
\noindent
\textbf{Notation and definitions:} \\
\noindent
Define
\[
G^*(u_0\mid\eta,\beta)
:=
\frac{1}{\Gamma(\eta)}
\int_0^{u_0}
\frac{t^{\eta-1}e^{-t/\beta}}{\beta^\eta}\,dt,
\qquad
h(\eta,\beta):=\log G^*(u_0\mid\eta,\beta),
\]
where \(u_0>0\) is the known threshold. Its first-order partial derivatives and the digamma function are defined, respectively, by
\[
h_\eta:=\frac{\partial h(\eta,\beta)}{\partial\eta},
\qquad
h_\beta:=\frac{\partial h(\eta,\beta)}{\partial\beta},
\qquad
\psi(\eta):=\frac{d}{d\eta}\log\Gamma(\eta).
\]
As we know, for \(X\geq u_0\), \(X\sim\mathrm{GPD}(\xi,\sigma,u_0)\), then $ Z:=\frac{X-u_0}{\sigma}\sim\mathrm{GPD}(\xi,1,0)$,
here \(\mathbb{E}_Z[\cdot]\) denotes expectation with respect to \(Z\).
\noindent
The corresponding score functions with respect to the model parameters are given by
\begin{equation*}
\frac{\partial\ell}{\partial\phi_1}=
\begin{cases}
\dfrac{1}{\phi_1},&x=0,\\
-\dfrac{1}{1-\phi_1-\phi_2},&0<x<u_0,\\
0,&x\geq u_0,
\end{cases}
\qquad
\frac{\partial\ell}{\partial\phi_2}=
\begin{cases}
0,&x=0,\\
-\dfrac{1}{1-\phi_1-\phi_2},&0<x<u_0,\\
\dfrac{1}{\phi_2},&x\geq u_0.
\end{cases}
\end{equation*}
Define notation $D_\eta:=\partial_\eta\log G^*(u_{0}\mid\eta,\beta)$ and
$D_\beta:=\partial_\beta\log G^*(u_{0}\mid\eta,\beta)$,
\[
\frac{\partial\ell}{\partial\eta}=
\begin{cases}
0,&x=0\text{ or }x\geq u_0,\\
\log(x/\beta)-\psi(\eta)-D_\eta,&0<x<u_0,
\end{cases}
\qquad
\frac{\partial\ell}{\partial\beta}=
\begin{cases}
0,&x=0\text{ or }x\geq u_0,\\
x/\beta^2-\eta/\beta-D_\beta,&0<x<u_0.
\end{cases}
\]
\begin{equation*}
\frac{\partial \ell}{\partial \sigma} =
\begin{cases}
-\dfrac{1}{\sigma}
+
\left(\dfrac{1}{\xi}+1\right)
\dfrac{\xi Z}{\sigma(1+\xi Z)},
& x \ge u_0, \\[8pt]
0, & \text{otherwise},
\end{cases}
\quad\quad
\frac{\partial \ell}{\partial \xi} =
\begin{cases}
\dfrac{1}{\xi^2}\log(1+\xi Z)
-\left(\dfrac{1}{\xi}+1\right)
\dfrac{Z}{1+\xi Z},
& x \ge u_0, \\[8pt]
0, & \text{otherwise}.
\end{cases}
\end{equation*}

One can easily verify that
\begin{align*}
    &\mathbb{E} \left[ \frac{\partial}{\partial \phi_1} \ell(X) \right] = 0, \quad
\mathbb{E} \left[ \frac{\partial}{\partial \phi_2} \ell(X) \right] = 0, \quad \mathbb{E} \left[ \frac{\partial}{\partial \eta} \ell(X) \right] = 0,  \\
& \mathbb{E} \left[ \frac{\partial}{\partial \beta} \ell(X) \right] = 0, \quad \mathbb{E} \left[ \frac{\partial}{\partial \xi} \ell(X) \right] = 0, \quad \text{and} \quad
\mathbb{E} \left[ \frac{\partial}{\partial \sigma} \ell(X) \right] = 0.
\end{align*}
\noindent The corresponding Fisher information matrix and its entries for \(\xi\ne0\) are presented in Section~\ref{asymptoticdistribution}.

\vspace{-0.4cm}
\section*{Score Functions: Case \(\xi = 0\)}
\noindent Let the parameter vector be $\boldsymbol{\theta} = (\phi_1, \eta, \beta, \sigma, \phi_2)$. For a single observation, the log-likelihood function is
\begin{equation}\label{eq:loglikelihood_single_xi0}
\small
\ell(\boldsymbol{\theta}; x) =
\begin{cases}
\log(\phi_1), & \text{if } x = 0, \\[6pt]

\log(1 - \phi_1 - \phi_2)
- \log G^*(u_{0} \mid \eta, \beta)
+ (\eta - 1)\log x
- \dfrac{x}{\beta} 
- \eta \log \beta
- \log \Gamma(\eta),
& \text{if } 0 < x < u_{0}, \\[8pt]

\log(\phi_2)
- \log \sigma
- \dfrac{x - u_{0}}{\sigma},
& \text{if } x \geq u_{0}.
\end{cases}
\end{equation}

\noindent For the case $\xi = 0$, the GPD tail reduces to an exponential distribution with mean $\sigma$. If $x \ge u_{0}$, the corresponding log-likelihood contribution is
\[
\ell_{\text{tail}}(\sigma; x) = \log(\phi_2) - \log(\sigma) - \frac{x - u_{0}}{\sigma}.
\]
\noindent The score function and its second derivative with respect to $\sigma$ are
\[
\frac{\partial\ell}{\partial\sigma}=
\begin{cases}
0,&x=0\text{ or }0<x<u_0,\\
-1/\sigma+(x-u_0)/\sigma^2,&x\geq u_0,
\end{cases}
\qquad
\frac{\partial^2\ell}{\partial\sigma^2}=
\begin{cases}
0,&x=0\text{ or }0<x<u_0,\\
1/\sigma^2-2(x-u_0)/\sigma^3,&x\geq u_0.
\end{cases}
\]
\noindent For \(X \ge u_{0}\), we have \(Y = X - u_{0} \sim \mathrm{Exp}(\sigma)\) with \(\mathbb{E}[Y] = \sigma\).  
 The corresponding Fisher information entries for \(\xi=0\) are presented in Section~\ref{asymptoticdistribution}.
\section*{Appendix A3. Supplementary Simulation Results}
\label{app:bias_mse_comparison_apendex}

\section*{Additional Bias and MSE Comparative Analysis:}
\noindent We considered the bias and MSE results for the parameter setting $\phi_1=0.4$, $\eta=1$, $\beta=5$, $u=11.5129$, $\xi=0.2$, $\sigma=5$, and $\phi_2=0.10$. The corresponding results are presented in Figure~\ref{fig: bias comparison2}. The estimates obtained from FEVIMM generally exhibit lower bias and MSE than those obtained from FEVMM and EVMM across the parameter estimates.
\section*{Additional Simulation Study Tables:}
\noindent Additional simulation study results are provided in Tables~\ref{Table-A2}--\ref{Table-A5}.

\begin{landscape}
\begin{figure}[!htbp]
		\centering
		\includegraphics[width=24.5cm,height=17.5cm]
{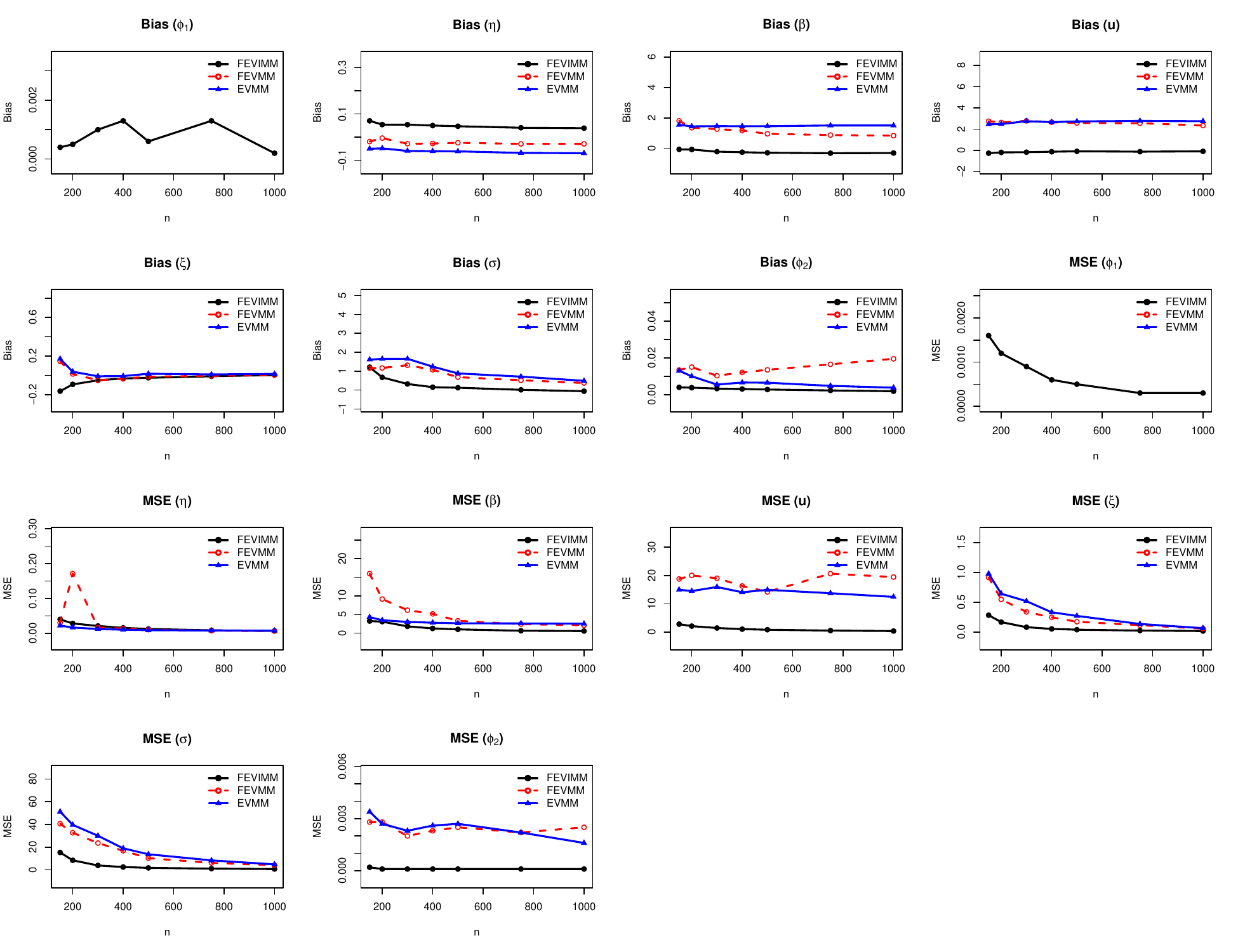}

		\caption{Bias and MSE comparison plots under Scenario 2, based on Monte Carlo replications ($N = 2000$): \(\phi_1 = 0.4\), \(\eta = 1\), \(\beta = 5\), \(u = 11.5129\), \(\xi = 0.2\), \(\sigma = 5\), \(\phi_2 = 0.10\).}
		\label{fig: bias comparison2}
   \end{figure} 
\end{landscape}
\noindent

\begin{landscape} 
\begin{table*}[htbp]
 \caption{Monte Carlo simulation results ($N=2000$) for sample sizes $n = 150, 200, 300, 400,$ and $500$, with an inlier proportion of 20\% and a heavy-tailed GPD component with a tail fraction of 10\%, reporting sample mean, BSE, 95\% BCI, MSE, bias, and CP.}
    \centering
    \scriptsize 
    \setlength{\tabcolsep}{3pt} 
    \renewcommand{\arraystretch}{1.3} 
    \begin{tabular}{|p{2.5cm}|p{1.8cm}|p{1.8cm}|p{2.4cm}|p{2cm}|p{3cm}|p{2.4cm}|p{2.3cm}|p{1.8cm}|} 
        \hline
         \textbf{$n$}  & \textbf{Parameter} & \textbf{True} & \textbf{Sample Mean} & \textbf{BSE} & \textbf{BCI} & \textbf{MSE} & \textbf{Bias} & \textbf{CP} \\ \hline
          \multirow{6}{*}{ 150 } & $\phi_1$ &  0.2  &  0.2  (NA) &  0.0415  &  (0.1073, 0.2696)  &  0.0011  (NA) &  0  (NA) &  95.8  \\ \cline{2-9}
                               & $\eta$   &  1  &  1.0518  ( 1.0108 ) &  0.2845  &  (0.8553, 1.9581)  &  0.0276  ( 0.0272 ) &  0.0518  ( 0.0108 ) &  98.2  \\ \cline{2-9}
                               & $\beta$  &  5  &  5.0049  ( 5.8235 ) &  2.6646  &  (2.1092, 11.1632)  &  2.7346  ( 7.4635 ) &  0.0049  ( 0.8235 ) &  98  \\ \cline{2-9}
                               & $u$      &  11.51293  &  11.6269  ( 12.5174 ) &  2.2359  &  (8.5447, 17.0922)  &  2.2209  ( 17.4559 ) &  0.114  ( 1.0044 ) &  95.6  \\ \cline{2-9}
                               & $\xi$    &  0.2  &  0.0721  ( 0.1635 ) &  4.6073  &  (-1.424, 13.06)  &  0.3103  ( 0.5303 ) &  -0.1279  ( -0.0365 ) &  100  \\ \cline{2-9}
                               & $\sigma$ &  5  &  5.9761  ( 6.3723 ) &  3.7671  &  (0, 12.3387)  &  15.0995  ( 32.1544 ) &  0.9761  ( 1.3723 ) &  90  \\ \cline{2-9}
                               & $\phi_2$ &  0.1  &  0.0994  ( 0.1117 ) &  0.0277  &  (0.039, 0.1541)  &  2e-04  ( 0.0026 ) &  -6e-04  ( 0.0117 ) &  99.6  \\ \hline

         \multirow{6}{*}{ 200 } & $\phi_1$ &  0.2  &  0.2001  (NA) &  0.0289  &  (0.1237, 0.2373)  &  8e-04  (NA) &  1e-04  (NA) &  96.6  \\ \cline{2-9}
                               & $\eta$   &  1  &  1.0481  ( 1.0079 ) &  0.1657  &  (0.887, 1.5353)  &  0.0208  ( 0.0187 ) &  0.0481  ( 0.0079 ) &  98.6  \\ \cline{2-9}
                               & $\beta$  &  5  &  4.9022  ( 5.5946 ) &  1.3283  &  (2.5239, 7.6142)  &  2.0874  ( 4.6312 ) &  -0.0978  ( 0.5946 ) &  98.2  \\ \cline{2-9}
                               & $u$      &  11.51293  &  11.6639  ( 12.7467 ) &  1.1737  &  (8.4, 13.1457)  &  1.6726  ( 13.1701 ) &  0.151  ( 1.2338 ) &  97.6  \\ \cline{2-9}
                               & $\xi$    &  0.2  &  0.1257  ( 0.1396 ) &  1.1812  &  (-0.2534, 4.3324)  &  0.1999  ( 0.3918 ) &  -0.0743  ( -0.0604 ) &  100  \\ \cline{2-9}
                               & $\sigma$ &  5  &  5.564  ( 6.116 ) &  1.6829  &  (0.5597, 6.7674)  &  9.4756  ( 20.4302 ) &  0.564  ( 1.116 ) &  92  \\ \cline{2-9}
                               & $\phi_2$ &  0.1  &  0.0994  ( 0.1071 ) &  0.0177  &  (0.0672, 0.1379)  &  1e-04  ( 0.0013 ) &  -6e-04  ( 0.0071 ) &  99.4  \\ \hline                  
          \multirow{6}{*}{ 300 } & $\phi_1$ &  0.2  &  0.2002  (NA) &  0.0261  &  (0.1489, 0.2525)  &  6e-04  (NA) &  2e-04  (NA) &  95.8  \\ \cline{2-9}
                               & $\eta$   &  1  &  1.0446  ( 1.0075 ) &  0.1539  &  (0.9338, 1.5535)  &  0.0145  ( 0.0126 ) &  0.0446  ( 0.0075 ) &  98.2  \\ \cline{2-9}
                               & $\beta$  &  5  &  4.8265  ( 5.3955 ) &  1.2127  &  (2.5525, 6.6039)  &  1.3063  ( 2.4968 ) &  -0.1735  ( 0.3955 ) &  97.2  \\ \cline{2-9}
                               & $u$      &  11.51293  &  11.6624  ( 12.6201 ) &  1.2702  &  (9.0593, 13.8905)  &  1.1542  ( 8.9011 ) &  0.1495  ( 1.1072 ) &  97.2  \\ \cline{2-9}
                               & $\xi$    &  0.2  &  0.1733  ( 0.141 ) &  0.9745  &  (-0.2417, 3.6749)  &  0.1115  ( 0.2416 ) &  -0.0267  ( -0.059 ) &  100  \\ \cline{2-9}
                               & $\sigma$ &  5  &  5.206  ( 5.8211 ) &  1.6659  &  (0.803, 7.1908)  &  5.3586  ( 12.6836 ) &  0.206  ( 0.8211 ) &  92.8  \\ \cline{2-9}
                               & $\phi_2$ &  0.1  &  0.0993  ( 0.1087 ) &  0.0173  &  (0.062, 0.1326)  &  1e-04  ( 0.0022 ) &  -7e-04  ( 0.0087 ) &  100  \\ \hline
                               
     \multirow{6}{*}{ 400 } & $\phi_1$ &  0.2  &  0.2004  (NA) &  0.0235  &  (0.1506, 0.2434)  &  4e-04  (NA) &  4e-04  (NA) &  97.6  \\ \cline{2-9}
                               & $\eta$   &  1  &  1.0439  ( 1.0058 ) &  0.1271  &  (0.8529, 1.3657)  &  0.0117  ( 0.0092 ) &  0.0439  ( 0.0058 ) &  98.4  \\ \cline{2-9}
                               & $\beta$  &  5  &  4.7763  ( 5.3123 ) &  1.3143  &  (2.8907, 7.7132)  &  0.9925  ( 2.0258 ) &  -0.2237  ( 0.3123 ) &  97.6  \\ \cline{2-9}
                               & $u$      &  11.51293  &  11.6765  ( 12.7026 ) &  1.5224  &  (8.6603, 14.3999)  &  0.9764  ( 14.3015 ) &  0.1636  ( 1.1897 ) &  97  \\ \cline{2-9}
                               & $\xi$    &  0.2  &  0.2046  ( 0.1671 ) &  0.9138  &  (-0.2739, 3.4286)  &  0.074  ( 0.1671 ) &  0.0046  ( -0.0329 ) &  100  \\ \cline{2-9}
                               & $\sigma$ &  5  &  4.994  ( 5.5325 ) &  1.8503  &  (0.9468, 8.1352)  &  3.283  ( 8.2301 ) &  -0.006  ( 0.5325 ) &  93  \\ \cline{2-9}
                               & $\phi_2$ &  0.1  &  0.0994  ( 0.109 ) &  0.0189  &  (0.0584, 0.1335)  &  1e-04  ( 0.0019 ) &  -6e-04  ( 0.009 ) &  99.6  \\ \hline
                               
       \multirow{6}{*}{ 500 } & $\phi_1$ &  0.2  &  0.1999  (NA) &  0.0218  &  (0.1675, 0.2552)  &  3e-04  (NA) &  -1e-04  (NA) &  97.4  \\ \cline{2-9}
                               & $\eta$   &  1  &  1.0437  ( 1.0056 ) &  0.1194  &  (0.9169, 1.3902)  &  0.0099  ( 0.0072 ) &  0.0437  ( 0.0056 ) &  97.8  \\ \cline{2-9}
                               & $\beta$  &  5  &  4.7329  ( 5.216 ) &  0.9617  &  (2.9098, 6.4437)  &  0.7909  ( 1.3264 ) &  -0.2671  ( 0.216 ) &  97.4  \\ \cline{2-9}
                               & $u$      &  11.51293  &  11.6952  ( 12.631 ) &  1.2178  &  (9.0939, 14.0366)  &  0.7769  ( 7.0034 ) &  0.1822  ( 1.1181 ) &  98  \\ \cline{2-9}
                               & $\xi$    &  0.2  &  0.2138  ( 0.1849 ) &  0.7159  &  (-0.2479, 2.8466)  &  0.0536  ( 0.1077 ) &  0.0138  ( -0.0151 ) &  100  \\ \cline{2-9}
                               & $\sigma$ &  5  &  4.919  ( 5.3463 ) &  1.426  &  (1.0752, 6.6847)  &  2.3895  ( 5.8333 ) &  -0.081  ( 0.3463 ) &  93  \\ \cline{2-9}
                               & $\phi_2$ &  0.1  &  0.0993  ( 0.1091 ) &  0.0176  &  (0.0616, 0.1302)  &  1e-04  ( 0.0015 ) &  -7e-04  ( 0.0091 ) &  99.6  \\ \hline
                               
    \end{tabular}
\captionsetup{
  labelfont=bf, 
  textfont=normal 
}
        \label{Table-A2}
\end{table*}
\end{landscape}

\begin{landscape} 
\begin{table*}[htbp]
 \caption{Monte Carlo simulation results ($N=2000$) for sample sizes $n = 150, 200, 300, 400,$ and $500$, with an inlier proportion of 20\% and a light-tailed GPD component with a tail fraction of 10\%, reporting sample mean, BSE, 95\% BCI, MSE, bias, and CP.}
    \centering
    \scriptsize 
    \setlength{\tabcolsep}{3pt} 
    \renewcommand{\arraystretch}{1.3} 
    \begin{tabular}{|p{2.5cm}|p{1.8cm}|p{1.8cm}|p{2.4cm}|p{2cm}|p{3cm}|p{2.4cm}|p{2.3cm}|p{1.8cm}|} 
        \hline
         \textbf{$n$}  & \textbf{Parameter} & \textbf{True} & \textbf{Sample Mean} & \textbf{BSE} & \textbf{BCI} & \textbf{MSE} & \textbf{Bias} & \textbf{CP} \\ \hline
         \multirow{6}{*}{ 150 } & $\phi_1$ &  0.2  &  0.1998  (NA) &  0.0421  &  (0.1095, 0.2774)  &  0.0012  (NA) &  -2e-04  (NA) &  96.2  \\ \cline{2-9}
                               & $\eta$   &  1  &  1.0626  ( 1.012 ) &  0.272  &  (0.8796, 1.9413)  &  0.0281  ( 0.0246 ) &  0.0626  ( 0.012 ) &  99.2  \\ \cline{2-9}
                               & $\beta$  &  5  &  4.8672  ( 5.7573 ) &  2.3268  &  (2.0606, 9.5827)  &  2.2091  ( 7.1443 ) &  -0.1328  ( 0.7573 ) &  98.8  \\ \cline{2-9}
                               & $u$      &  11.51293  &  11.6441  ( 12.4537 ) &  2.0997  &  (8.5244, 16.0821)  &  2.0259  ( 4.8627 ) &  0.1312  ( 0.9408 ) &  96.2  \\ \cline{2-9}
                               & $\xi$    &  -0.2  &  -0.3908  ( -0.2112 ) &  3.4169  &  (-1.7481, 9.4345)  &  0.2678  ( 0.3014 ) &  -0.1908  ( -0.0112 ) &  100  \\ \cline{2-9}
                               & $\sigma$ &  5  &  6.0651  ( 5.1536 ) &  2.9526  &  (0.0487, 9.8235)  &  11.6052  ( 11.4555 ) &  1.0651  ( 0.1536 ) &  93.2  \\ \cline{2-9}
                               & $\phi_2$ &  0.1  &  0.0991  ( 0.109 ) &  0.0293  &  (0.0363, 0.1564)  &  2e-04  ( 0.0017 ) &  -9e-04  ( 0.009 ) &  99.6  \\ \hline

          \multirow{6}{*}{ 200 } & $\phi_1$ &  0.2  &  0.2002  (NA) &  0.0332  &  (0.1281, 0.2579)  &  9e-04  (NA) &  2e-04  (NA) &  97.6  \\ \cline{2-9}
                               & $\eta$   &  1  &  1.0528  ( 1.0087 ) &  0.1907  &  (0.8336, 1.5823)  &  0.0206  ( 0.0184 ) &  0.0528  ( 0.0087 ) &  99.4  \\ \cline{2-9}
                               & $\beta$  &  5  &  4.8545  ( 5.5741 ) &  2.8094  &  (2.545, 10.847)  &  1.6704  ( 5.0616 ) &  -0.1455  ( 0.5741 ) &  99  \\ \cline{2-9}
                               & $u$      &  11.51293  &  11.6823  ( 12.3912 ) &  1.8163  &  (8.0953, 15.0391)  &  1.5629  ( 3.7971 ) &  0.1693  ( 0.8783 ) &  97.8  \\ \cline{2-9}
                               & $\xi$    &  -0.2  &  -0.332  ( -0.2753 ) &  2.4009  &  (-1.0478, 8.463)  &  0.1633  ( 0.208 ) &  -0.132  ( -0.0753 ) &  100  \\ \cline{2-9}
                               & $\sigma$ &  5  &  5.6843  ( 5.3701 ) &  2.7666  &  (0.2449, 10.541)  &  7.2562  ( 8.723 ) &  0.6843  ( 0.3701 ) &  93.6  \\ \cline{2-9}
                               & $\phi_2$ &  0.1  &  0.0982  ( 0.1088 ) &  0.0237  &  (0.046, 0.1422)  &  2e-04  ( 0.0013 ) &  -0.0018  ( 0.0088 ) &  100  \\ \hline                  
        \multirow{6}{*}{ 300 } & $\phi_1$ &  0.2  &  0.2006  (NA) &  0.0296  &  (0.1578, 0.2743)  &  6e-04  (NA) &  6e-04  (NA) &  97.6  \\ \cline{2-9}
                               & $\eta$   &  1  &  1.0487  ( 1.0047 ) &  0.1595  &  (0.82, 1.4508)  &  0.0149  ( 0.0119 ) &  0.0487  ( 0.0047 ) &  99  \\ \cline{2-9}
                               & $\beta$  &  5  &  4.7834  ( 5.4332 ) &  3.0101  &  (2.8715, 10.989)  &  1.1006  ( 2.8143 ) &  -0.2166  ( 0.4332 ) &  97.8  \\ \cline{2-9}
                               & $u$      &  11.51293  &  11.6509  ( 12.4659 ) &  1.7304  &  (8.7308, 15.2942)  &  1.0684  ( 3.2094 ) &  0.138  ( 0.953 ) &  98.6  \\ \cline{2-9}
                               & $\xi$    &  -0.2  &  -0.277  ( -0.2813 ) &  1.5561  &  (-1.0165, 5.7825)  &  0.0856  ( 0.1309 ) &  -0.077  ( -0.0813 ) &  100  \\ \cline{2-9}
                               & $\sigma$ &  5  &  5.3187  ( 5.2817 ) &  2.416  &  (0.5528, 9.637)  &  3.8239  ( 5.936 ) &  0.3187  ( 0.2817 ) &  95.2  \\ \cline{2-9}
                               & $\phi_2$ &  0.1  &  0.0988  ( 0.1065 ) &  0.0232  &  (0.0462, 0.137)  &  1e-04  ( 9e-04 ) &  -0.0012  ( 0.0065 ) &  99.8  \\ \hline
                               
        \multirow{6}{*}{ 400 } & $\phi_1$ &  0.2  &  0.2007  (NA) &  0.0262  &  (0.1715, 0.2749)  &  4e-04  (NA) &  7e-04  (NA) &  98.6  \\ \cline{2-9}
                               & $\eta$   &  1  &  1.0449  ( 1.0054 ) &  0.1264  &  (0.8043, 1.2984)  &  0.0114  ( 0.0088 ) &  0.0449  ( 0.0054 ) &  99.4  \\ \cline{2-9}
                               & $\beta$  &  5  &  4.762  ( 5.2846 ) &  1.927  &  (3.0897, 10.0069)  &  0.8779  ( 1.6176 ) &  -0.238  ( 0.2846 ) &  97.8  \\ \cline{2-9}
                               & $u$      &  11.51293  &  11.6893  ( 12.4709 ) &  1.4436  &  (8.4334, 14.0623)  &  0.852  ( 2.7865 ) &  0.1764  ( 0.9579 ) &  98.4  \\ \cline{2-9}
                               & $\xi$    &  -0.2  &  -0.2443  ( -0.2515 ) &  1.0542  &  (-0.7049, 3.9204)  &  0.0525  ( 0.0877 ) &  -0.0443  ( -0.0515 ) &  100  \\ \cline{2-9}
                               & $\sigma$ &  5  &  5.1052  ( 5.0683 ) &  1.9057  &  (0.7028, 7.9816)  &  2.3763  ( 3.9096 ) &  0.1052  ( 0.0683 ) &  96.4  \\ \cline{2-9}
                               & $\phi_2$ &  0.1  &  0.0982  ( 0.1065 ) &  0.0224  &  (0.0472, 0.1361)  &  1e-04  ( 7e-04 ) &  -0.0018  ( 0.0065 ) &  99.8  \\ \hline
                               
          \multirow{6}{*}{ 500 } & $\phi_1$ &  0.2  &  0.2005  (NA) &  0.0226  &  (0.1634, 0.2518)  &  4e-04  (NA) &  5e-04  (NA) &  98  \\ \cline{2-9}
                               & $\eta$   &  1  &  1.0439  ( 1.0047 ) &  0.1191  &  (0.893, 1.3532)  &  0.0098  ( 0.0069 ) &  0.0439  ( 0.0047 ) &  98.8  \\ \cline{2-9}
                               & $\beta$  &  5  &  4.7361  ( 5.2098 ) &  1.0433  &  (3.0119, 6.9197)  &  0.7195  ( 1.1523 ) &  -0.2639  ( 0.2098 ) &  97.8  \\ \cline{2-9}
                               & $u$      &  11.51293  &  11.7041  ( 12.5128 ) &  1.1391  &  (9.7273, 14.3337)  &  0.7229  ( 2.5916 ) &  0.1912  ( 0.9999 ) &  98.6  \\ \cline{2-9}
                               & $\xi$    &  -0.2  &  -0.231  ( -0.2449 ) &  0.676  &  (-0.542, 2.3146)  &  0.0375  ( 0.0621 ) &  -0.031  ( -0.0449 ) &  100  \\ \cline{2-9}
                               & $\sigma$ &  5  &  5.0279  ( 4.983 ) &  1.2029  &  (1.001, 6.0445)  &  1.7597  ( 2.9304 ) &  0.0279  ( -0.017 ) &  96  \\ \cline{2-9}
                               & $\phi_2$ &  0.1  &  0.0981  ( 0.1058 ) &  0.0183  &  (0.0532, 0.1255)  &  1e-04  ( 5e-04 ) &  -0.0019  ( 0.0058 ) &  100  \\ \hline
                               
    \end{tabular}
\captionsetup{
  labelfont=bf, 
  textfont=normal 
}
       \label{Table-A3}
\end{table*}
\end{landscape}

\begin{landscape} 
\begin{table*}[htbp]
 \caption{Monte Carlo simulation results ($N=2000$) for sample sizes $n = 150, 200, 300, 400,$ and $500$, with an inlier proportion of 20\% and an exponential-tailed GPD component with a tail fraction of 10\%, reporting sample mean, BSE, 95\% BCI, MSE, bias, and CP.}
    \centering
    \scriptsize 
    \setlength{\tabcolsep}{3pt} 
    \renewcommand{\arraystretch}{1.3} 
    \begin{tabular}{|p{2.5cm}|p{1.8cm}|p{1.8cm}|p{2.4cm}|p{2cm}|p{3cm}|p{2.4cm}|p{2.3cm}|p{1.8cm}|} 
        \hline
         \textbf{$n$}  & \textbf{Parameter} & \textbf{True} & \textbf{Sample Mean} & \textbf{BSE} & \textbf{BCI} & \textbf{MSE} & \textbf{Bias} & \textbf{CP} \\ \hline
       \multirow{6}{*}{ 150 } & $\phi_1$ &  0.2  &  0.1995  (NA) &  0.0411  &  (0.1009, 0.2629)  &  0.0011  (NA) &  -5e-04  (NA) &  95.2  \\ \cline{2-9}
                               & $\eta$   &  1  &  1.0566  ( 1.01 ) &  0.2881  &  (0.8717, 1.9848)  &  0.0277  ( 0.025 ) &  0.0566  ( 0.01 ) &  98.8  \\ \cline{2-9}
                               & $\beta$  &  5  &  4.9482  ( 5.8036 ) &  2.5303  &  (2.1044, 10.272)  &  2.5789  ( 7.1984 ) &  -0.0518  ( 0.8036 ) &  98  \\ \cline{2-9}
                               & $u$      &  11.51293  &  11.6241  ( 12.4407 ) &  2.2811  &  (8.7312, 17.0897)  &  2.1218  ( 6.4938 ) &  0.1112  ( 0.9278 ) &  96  \\ \cline{2-9}
                               & $\xi$    &  0  &  -0.1567  ( -0.0344 ) &  4.1787  &  (-1.6492, 11.4489)  &  0.279  ( 0.3991 ) &  -0.1567  ( -0.0344 ) &  100  \\ \cline{2-9}
                               & $\sigma$ &  5  &  6.0142  ( 5.7533 ) &  3.491  &  (0.0036, 11.4062)  &  13.0908  ( 17.9915 ) &  1.0142  ( 0.7533 ) &  96.6  \\ \cline{2-9}
                               & $\phi_2$ &  0.1  &  0.0998  ( 0.1111 ) &  0.0296  &  (0.0355, 0.1574)  &  2e-04  ( 0.0029 ) &  -2e-04  ( 0.0111 ) &  99.6  \\ \hline

         \multirow{6}{*}{ 200 } & $\phi_1$ &  0.2  &  0.2  (NA) &  0.0304  &  (0.129, 0.249)  &  9e-04  (NA) &  0  (NA) &  96.6  \\ \cline{2-9}
                               & $\eta$   &  1  &  1.0501  ( 1.0077 ) &  0.1704  &  (0.863, 1.5384)  &  0.0211  ( 0.0181 ) &  0.0501  ( 0.0077 ) &  98.8  \\ \cline{2-9}
                               & $\beta$  &  5  &  4.8911  ( 5.616 ) &  1.659  &  (2.4762, 7.7705)  &  1.9265  ( 5.4137 ) &  -0.1089  ( 0.616 ) &  97.4  \\ \cline{2-9}
                               & $u$      &  11.51293  &  11.643  ( 12.5278 ) &  1.2965  &  (8.3599, 13.4926)  &  1.6105  ( 5.3347 ) &  0.1301  ( 1.0149 ) &  97  \\ \cline{2-9}
                               & $\xi$    &  0  &  -0.1059  ( -0.0767 ) &  1.3642  &  (-0.5237, 4.9906)  &  0.1806  ( 0.301 ) &  -0.1059  ( -0.0767 ) &  100  \\ \cline{2-9}
                               & $\sigma$ &  5  &  5.6523  ( 5.7803 ) &  1.7828  &  (0.4876, 7.1495)  &  8.487  ( 13.7407 ) &  0.6523  ( 0.7803 ) &  92.8  \\ \cline{2-9}
                               & $\phi_2$ &  0.1  &  0.0991  ( 0.1077 ) &  0.0188  &  (0.0601, 0.1354)  &  1e-04  ( 0.0013 ) &  -9e-04  ( 0.0077 ) &  99.6  \\ \hline                  
        \multirow{6}{*}{ 300 } & $\phi_1$ &  0.2  &  0.2002  (NA) &  0.0282  &  (0.1569, 0.2671)  &  6e-04  (NA) &  2e-04  (NA) &  96.8  \\ \cline{2-9}
                               & $\eta$   &  1  &  1.0454  ( 1.0066 ) &  0.1574  &  (0.8401, 1.4436)  &  0.0145  ( 0.0122 ) &  0.0454  ( 0.0066 ) &  98.4  \\ \cline{2-9}
                               & $\beta$  &  5  &  4.8197  ( 5.4097 ) &  2.2739  &  (2.8882, 10.1398)  &  1.2771  ( 2.8021 ) &  -0.1803  ( 0.4097 ) &  97.4  \\ \cline{2-9}
                               & $u$      &  11.51293  &  11.6453  ( 12.5993 ) &  1.5487  &  (8.7022, 14.6824)  &  1.1217  ( 5.2309 ) &  0.1323  ( 1.0864 ) &  98  \\ \cline{2-9}
                               & $\xi$    &  0  &  -0.0513  ( -0.0816 ) &  1.4302  &  (-0.6094, 5.0821)  &  0.0948  ( 0.1919 ) &  -0.0513  ( -0.0816 ) &  100  \\ \cline{2-9}
                               & $\sigma$ &  5  &  5.2412  ( 5.6228 ) &  2.2727  &  (0.5872, 9.1487)  &  4.3657  ( 9.0232 ) &  0.2412  ( 0.6228 ) &  94  \\ \cline{2-9}
                               & $\phi_2$ &  0.1  &  0.0994  ( 0.1059 ) &  0.0219  &  (0.053, 0.1403)  &  1e-04  ( 0.0012 ) &  -6e-04  ( 0.0059 ) &  99.8  \\ \hline
     \multirow{6}{*}{ 400 } & $\phi_1$ &  0.2  &  0.2007  (NA) &  0.0254  &  (0.169, 0.2671)  &  4e-04  (NA) &  7e-04  (NA) &  98  \\ \cline{2-9}
                               & $\eta$   &  1  &  1.0441  ( 1.0076 ) &  0.1266  &  (0.8246, 1.3189)  &  0.0116  ( 0.009 ) &  0.0441  ( 0.0076 ) &  100  \\ \cline{2-9}
                               & $\beta$  &  5  &  4.7766  ( 5.2612 ) &  2.0019  &  (3.165, 9.5729)  &  0.9569  ( 1.7566 ) &  -0.2234  ( 0.2612 ) &  97  \\ \cline{2-9}
                               & $u$      &  11.51293  &  11.6556  ( 12.5106 ) &  1.4359  &  (8.6198, 14.1305)  &  0.902  ( 3.7987 ) &  0.1426  ( 0.9977 ) &  98  \\ \cline{2-9}
                               & $\xi$    &  0  &  -0.0215  ( -0.039 ) &  0.8692  &  (-0.5004, 3.0791)  &  0.0614  ( 0.1139 ) &  -0.0215  ( -0.039 ) &  100  \\ \cline{2-9}
                               & $\sigma$ &  5  &  5.0492  ( 5.2699 ) &  1.7197  &  (0.857, 7.4853)  &  2.8166  ( 5.6663 ) &  0.0492  ( 0.2699 ) &  94  \\ \cline{2-9}
                               & $\phi_2$ &  0.1  &  0.0991  ( 0.1079 ) &  0.0214  &  (0.0525, 0.1376)  &  1e-04  ( 0.0012 ) &  -9e-04  ( 0.0079 ) &  100  \\ \hline
                               
       \multirow{6}{*}{ 500 } & $\phi_1$ &  0.2  &  0.2003  (NA) &  0.0221  &  (0.1641, 0.2507)  &  4e-04  (NA) &  3e-04  (NA) &  98  \\ \cline{2-9}
                               & $\eta$   &  1  &  1.0422  ( 1.0059 ) &  0.121  &  (0.8908, 1.3629)  &  0.0097  ( 0.0072 ) &  0.0422  ( 0.0059 ) &  96  \\ \cline{2-9}
                               & $\beta$  &  5  &  4.7497  ( 5.2037 ) &  1.2757  &  (3.0166, 6.9628)  &  0.7791  ( 1.3187 ) &  -0.2503  ( 0.2037 ) &  94  \\ \cline{2-9}
                               & $u$      &  11.51293  &  11.6806  ( 12.5624 ) &  1.2036  &  (9.4116, 14.2768)  &  0.751  ( 4.6412 ) &  0.1677  ( 1.0494 ) &  98  \\ \cline{2-9}
                               & $\xi$    &  0  &  -0.0058  ( -0.033 ) &  0.6849  &  (-0.3969, 2.5912)  &  0.0445  ( 0.0888 ) &  -0.0058  ( -0.033 ) &  100  \\ \cline{2-9}
                               & $\sigma$ &  5  &  4.9603  ( 5.2079 ) &  1.3326  &  (1.0623, 6.3656)  &  2.1219  ( 4.4745 ) &  -0.0397  ( 0.2079 ) &  96.5  \\ \cline{2-9}
                               & $\phi_2$ &  0.1  &  0.099  ( 0.1082 ) &  0.0177  &  (0.0574, 0.1272)  &  1e-04  ( 0.0016 ) &  -0.001  ( 0.0082 ) &  100  \\ \hline
                               
    \end{tabular}
\captionsetup{
  labelfont=bf, 
  textfont=normal 
}
        \label{Table-A4}
\end{table*}
\end{landscape}

\vspace{-0.5cm}\begin{table}[!htbp]
\caption{Parameter estimates obtained from the proposed FEVIMM and classical methods (MEP, PSP, and PP) for data with and without inliers.}
\centering
\small
\setlength{\tabcolsep}{4pt} 
\renewcommand{\arraystretch}{1.0}
\begin{tabular}{lcccccccc}
\hline
\textbf{Parameter} 
& \textbf{True}  & \textbf{FEVIMM} 
& \textbf{MEP ($\mathbf{x}$)} & \textbf{MEP ($\mathbf{x}^{+}$)}    
& \textbf{PSP ($\mathbf{x}$)} & \textbf{PSP ($\mathbf{x}^{+}$)} 
& \textbf{PP ($\mathbf{x}$)}  & \textbf{PP ($\mathbf{x}^{+}$)} \\ \hline

\(\phi_1\)  & 0.300   & 0.284     & N/A & N/A & N/A & N/A & N/A & N/A \\
\(\eta\)    & 1.000   & 0.904     & N/A & N/A & N/A & N/A & N/A & N/A \\
\(\beta\)   & 5.000   & 6.343     & N/A & N/A & N/A & N/A & N/A & N/A \\
\(u\)       & 9.486  & 11.553    &   12.000  & 13.000 & 11.000  & 12.000 & 12.000 & 13.000 \\
\(\sigma\)  & 5.000   & 4.942     & 5.400 & 6.500 & 5.400 & 6.500 & N/A & N/A\\
\(\xi\)     & 0.200   & 0.190     & 0.140 & 0.044 & 0.140 & 0.044 & -0.001 & -0.400 \\
\(\phi_2\)  & 0.150   & 0.107     & N/A & N/A & N/A & N/A & N/A & N/A \\ \hline

\end{tabular}
    \label{Table-A5}
\end{table}

\section*{Appendix A4. Summary of EVMMs in the Literature}
\label{Appendix_EVMM_summary}
\noindent A summary of some of EVMMs from the literature is provided in Table~\ref{tab:evmm_lit}.
\begin{table}[!htbp]
\centering
\caption{Summary of EVMMs in the literature}
\resizebox{1\textwidth}{!}{%
\begin{tabular}{@{} l l l @{}}
\toprule
\textbf{Author(s)} & \textbf{Bulk distribution} & \textbf{Tail distribution} \\
\midrule
\citet{behrens2004bayesian} & Gamma, Weibull & GPD \\
\citet{tancredi2006accounting} & Mixture of uniform & GPD \\
\citet{macdonald2012extreme} & Kernel density estimator & GPD \\
\citet{do2012semiparametric} & Mixture of gamma & GPD \\
\citet{hu2018evmix} & Boundary corrected kernel density estimator & GPD \\
\citet{marambakuyana2024composite} & Various bulk & Various tail \\

\bottomrule
\end{tabular}%
}
\label{tab:evmm_lit}
\end{table}

\end{document}